\documentclass[10pt, journal]{IEEEtran}

\usepackage{graphics}

\usepackage{lipsum}
\usepackage{tablefootnote}
\usepackage{cite}

\usepackage{url}

\usepackage{amsthm} 
\usepackage{amsmath}    
\IEEEoverridecommandlockouts
\usepackage{empheq}
\usepackage{bm}
\usepackage{soul}
\usepackage{graphicx}
\usepackage{indentfirst}
\usepackage{epstopdf}
\usepackage{amssymb}
\usepackage{url}
\usepackage{enumitem} 
\usepackage{multirow}
\usepackage{hhline}
\usepackage{booktabs}
\usepackage{mathtools}
\usepackage{makecell}
\usepackage[linesnumbered,boxed,commentsnumbered,ruled,vlined,longend]{algorithm2e}
\usepackage{comment}

\DeclareMathOperator*{\minimize}{minimize}

\DeclareMathOperator*{\subjectto}{subject\ to}

\DeclareMathOperator{\Blkdiag}{Blkdiag}
\DeclareMathOperator{\Diag}{Diag}
\makeatother
\DeclareMathAlphabet\mathbfcal{OMS}{cmsy}{b}{n}

\newtheorem{theorem}{Theorem}
\newtheorem{mydef}{Definition}



\usepackage{stackengine}

\newcommand{\mat}[1]{\boldsymbol{#1}}

\newcommand{\bmat}[1]{\begin{bmatrix} #1 \end{bmatrix}}

\providecommand{\mA}{\ensuremath{\mat{A}}}
\providecommand{\mB}{\ensuremath{\mat{B}}}
\providecommand{\mC}{\ensuremath{\mat{C}}}
\providecommand{\mD}{\ensuremath{\mat{D}}}
\providecommand{\mE}{\ensuremath{\mat{E}}}
\providecommand{\mF}{\ensuremath{\mat{F}}}
\providecommand{\mG}{\ensuremath{\mat{G}}}

\providecommand{\mI}{\ensuremath{\mat{I}}}

\providecommand{\mL}{\ensuremath{\mat{L}}}
\providecommand{\mM}{\ensuremath{\mat{M}}}
\providecommand{\mN}{\ensuremath{\mat{N}}}
\providecommand{\mO}{\ensuremath{\mat{O}}}
\providecommand{\mP}{\ensuremath{\mat{P}}}
\providecommand{\mQ}{\ensuremath{\mat{Q}}}
\providecommand{\mR}{\ensuremath{\mat{R}}}

\providecommand{\mX}{\ensuremath{\mat{X}}}
\providecommand{\mY}{\ensuremath{\mat{Y}}}
\providecommand{\mZ}{\ensuremath{\mat{Z}}}

\newcommand{\m}{\boldsymbol}
\allowdisplaybreaks[4]
\usepackage[colorlinks = true,
linkcolor = blue,
urlcolor  = blue,
citecolor = blue,
anchorcolor = blue]{hyperref}

\newcommand{\mc}[1]{\mathcal{#1}}
\newcommand{\mbb}[1]{\mathbb{#1}}
\newcommand{\mr}[1]{\mathrm{#1}}
\usepackage[framemethod=TikZ]{mdframed}
\mdfdefinestyle{MyFrame}{%
	linecolor=black,
	outerlinewidth=1.25pt,
	roundcorner=1.25pt,
	innerrightmargin=5pt,
	innerleftmargin=5pt,}
	
\usepackage[noabbrev]{cleveref}

\usepackage{mathtools}

\DeclarePairedDelimiter\abs{\lvert}{\rvert}%
\DeclarePairedDelimiter\norm{\lVert}{\rVert}%

\makeatletter
\let\oldabs\abs
\def\abs{\@ifstar{\oldabs}{\oldabs*}}
\let\oldnorm\norm
\def\norm{\@ifstar{\oldnorm}{\oldnorm*}}
\makeatother

\usepackage[english]{babel}
\usepackage[utf8]{inputenc}
\usepackage[super]{nth}

\usepackage{float}
\usepackage[caption = false]{subfig}

\usepackage{array}
\usepackage{threeparttable}

\renewcommand\IEEEkeywordsname{Keywords}

\usepackage[english]{babel}
\usepackage[utf8]{inputenc}
\usepackage[super]{nth}

\RequirePackage{filecontents}

\SetKwRepeat{Do}{do}{while}%

\title{\centering \Huge {Dynamic State Estimation of Nonlinear Differential Algebraic Equation Models of Power Networks}}
\author{Muhammad Nadeem$^\star$, Sebastian A. Nugroho$^\dagger$, and Ahmad F. Taha$^\star$
\thanks{$^\star$Civil and Environmental Engineering Department, Vanderbilt University, 2201 West End Ave, Nashville, Tennessee 37235. $^\dagger$Department of Electrical Engineering and Computer Science University of Michigan, Ann Arbor
Emails: muhammad.nadeem@vanderbilt.edu, ahmad.taha@vanderbilt.edu, snugroho@umich.edu. This work is supported by National Science Foundation Grants 2151571 and 2152450.}
}

\begin{document}

\newdimen\origiwspc%
\newdimen\origiwstr%
\origiwspc=\fontdimen2\font
\origiwstr=\fontdimen3\font

\fontdimen2\font=0.63ex

\maketitle
\markboth{IEEE Transactions on Power Systems, In Press, June 2022}{}

\begin{abstract}
This paper investigates the joint problems of dynamic state estimation of algebraic variables (voltage and phase angle) and generator states (rotor angle and frequency) of nonlinear differential algebraic equation (NDAE) power network models, under uncertainty. Traditionally, these two problems have been decoupled due to complexity of handling NDAE models.  In particular, this paper offers the first attempt to solve the aforementioned problem in a coupled approach where the algebraic and generator states estimates are simultaneously computed. The proposed estimation algorithm herein is endowed with the following properties: \textit{(i)} it is fairly simple to implement and based on well-understood Lyapunov theory; \textit{(ii)} considers various sources of uncertainty from generator control inputs, loads, renewables, process and measurement noise; \textit{(iii)} models phasor measurement unit installations at arbitrary buses; and \textit{(iv)} is computationally less intensive than the decoupled approach in the literature.  
 \end{abstract}


\section{Introduction}\label{section:intro}
\IEEEPARstart{A}{s} human-made climate change is necessitating increased penetration of fuel-free energy resources, monitoring and realtime control of the power grid transients have become more complex. This is due to the fact that controlling and predicting mostly non-dispatchable intermittent and uncertain renewable energy (e.g., solar and wind) results in a more challenging control problem of dispatchable generators (e.g., synchronous machines) in the short, second-to-second time-scale. 

An essential component of feedback control in power networks is dynamic state estimation (DSE) which accurately estimates the grid's physical states in realtime. Typically, DSE uses two essential components. The first component is a physics-based model of the electromechanical transients depicted via a set of nonlinear differential algebraic equations ~\cite{sauer2017power}
\begin{subequations}~\label{equ:PSModel}
	\begin{align}
\textit{machine dynamics:} \;\;\;\;\;	\dot{\m x}(t) &= \m f(\m x, \m a, \m u, \m w)  ~\label{equ:PSModel-a} \\
\textit{algebric constraints:} \;\;\;\;\;	\m 0 &= \m h(\m x, \m a, \m w). ~\label{equ:PSModel-b}
	\end{align}
\end{subequations}
In~\eqref{equ:PSModel}, vectors $\m x(t), \m a(t), \m u(t)$, and $\m w(t)$ depict the states of all generators (e.g., angle and frequency), algebraic variables of all nodes (e.g., voltages and power flows), controllable inputs  (e.g., exciter field voltage and mechanical input power), and the uncontrollable inputs from loads, renewables, and network's parameters respectively. The differential equations~\eqref{equ:PSModel-a} model the transients of the generators while the algebraic equations~\eqref{equ:PSModel-b} depict the network power flow constraints. Vintage power system textbooks~\cite{sauer2017power,kundur2007power} delineate this NDAE model in great detail while showcasing the vector-valued nonlinearities abstracted through $\m f(\cdot)$ and $\m h(\cdot)$.

The second component needed in DSE is a stream of data representing high sampling rate, synchronized measurements $\m y(t)$ from phasor measurement units (PMUs). These measurements can be modeled as
\begin{equation}~\label{equ:PMUModel}
	\m y(t) = \m g(\m x, \m a, \m v)
\end{equation} 
where $\m g(\cdot)$ is a nonlinear vector-valued function that maps both the dynamic and algebraic variables to what PMUs typically measure and $\m v(t)$ encapsulates measurement noise.
 
To that end, it is natural to ask about how to perform dynamic state estimation for the complete model of a power system depicted in \eqref{equ:PSModel} using PMUs measurement model \eqref{equ:PMUModel} without jeopardizing system accuracy. In this paper we propose a unique way by developing a simple method to estimate both algebraic and dynamics state of power systems simultaneously via well-understood Lyapunov and control theory. 

Extensive research has been carried out in the past two decades to perform DSE by focusing mainly on generator dynamic models and ignoring algebraic constraints and states. That is, the bulk of the literature investigates DSE based on ODE models, rather than a complete NDAE one. Most of the developed algorithms are either stochastic estimators (e.g., Kalman filters and its derivatives) or deterministic observers. The recent survey papers \cite{Liu2021TPWRS,8624411} produce a thorough summary of the state-of-the-art in power system DSE.

In stochastic estimators category, extended Kalman filter (EKF) \cite{ZhenyuIPEC2007}, unscented Kalman filter (UKF) \cite{GhahremaniITPWRS2011}, particle filter (PF) \cite{CuiITWPRS2015}, extended particle filter (EPF)  \cite{ZhouITPWRS2013} and ensemble Kalman filter (EnKF) \cite{6281499}  have been developed. A thorough comparative study has also been carried out between EnKF, UKF, PF, and EKF in \cite{ZhouITSG2015, LiuIEEACC2020} and the different strengths and weaknesses of each estimator have been highlighted.

To deal with process and measurement noise, model uncertainties and unknown inputs, robust versions of these stochastic estimators have also been developed such as, EKF with unknown inputs (EKF-UI) \cite{GhahremaniITWPRS2016}, generalized maximum likelihood EKF (GM-IEKF)  \cite{ZhaoITWPRS2017}, $H_\infty$ based EKF \cite{ZhaoITWPRS2018}, robust UKF \cite{ZhaoITSG2019} and robust CKF \cite{LiIEEACC2019} have been proposed.

As for deterministic observers, the core idea is to drive the error between original and estimated states to asymptotically zero (or bounded region around the origin) via convex optimization formulations, which are derived using Lyapunov stability criteria.  For example, the authors in \cite{SebastianITPWRS2020} design an observer for the ODE model of power system while considering the nonlinearities as Lipschitz bounded. In \cite{Jin2018MultiplierbasedOD}, a multiplier based observer design has been proposed and tested on a multi-machine power system with Lipschitz nonlinearities. In \cite{HaesITSG2020}, an unknown input observer has been designed for a linearized model of a power system. 

Although the NDAE and PMUs measurement  models~\eqref{equ:PSModel}--\eqref{equ:PMUModel} abstract grid transients adequately, the literature of DSE in power systems DSE either ignore the algebraic constraints (or reduce them), or linearize the  nonlinearities  around a certain operating point while ignoring algebraic power flow constraints.
In recent  studies, different methodologies have been proposed to perform linearization free DSE by dealing with nonlinearities through unscented transformation \cite{ZhaoITSG2019} and by considering them to be Lipschitz bounded \cite{SebastianITPWRS2020}. However, algebraic power flow constraints have been completely neglected mainly due to the complexity in performing DSE for the complete NDAE model. 

In most of the DSE literature the NDAE model is oversimplified usually through Kron reduction to obtain a nonlinear ODE model of a power system which can then be tackled using a plethora of DSE algorithms as discussed in the previous section. However, ODE models forego algebraic constraints and thus cannot capture system's dynamics pertaining to topological changes (e.g., sudden tripping of transmission lines) \cite{GRO201612}. Also, it is not clear how to incorporate loads and renewables in the nonlinear ODE models and subsequent DSE routines \cite{Wu2019InfluenceOL}. Another disadvantage is that, if ODE models are solely used for DSE, then PMU locations are limited to generator buses only.

\textcolor{black}{Thus, and as an alternative to performing DSE for the NDAE model, a common approach is to first estimate the algebraic variables of a power system and then use these values to estimate generator internal states using KF or its derivatives---a decoupled two-step approach. Such as in \cite{RouhaniITSG2018} a least absolute value (LAV)-based estimator is first used to estimate algebraic variables and then UKF is used in the second stage to estimate dynamic states of the power system. In \cite{ZhangITSE2014}, bus voltage and phase angle of the power system are first estimated using raw PMU measurements through adaptive KF and then these estimated algebraic variables along with EKF is used to estimate the generator dynamic states. However, a simple linearized second-order generator model has been used and the algebraic constraints \eqref{equ:PSModel-b} related to the NDAE model have been completely neglected. Similarly in \cite{RinaldiINPROCEEDINGS2017} a combination of Newton-Raphson based and higher mode sliding observer is used to estimate both algebraic and dynamic variables.}

\textcolor{black}{We also want to point out here that observability of a power system is generally a prerequisite while performing DSE. For linear time-invariant (LTI) model of power systems the observability can easily be assessed by checking the rank of the observability matrix \cite{8624411}. However for nonlinear models the observability depends on the operating point, thus a system may move between strongly and weakly observable for given measurements with changing operating/equilibrium points. Extensive research has also been carried out to assess the observability of a nonlinear model of a power system. For example in \cite{RouhaniITPWRS2017} lie derivative based method has been proposed to check the observability of fourth order synchronous generator model. It has been shown that measuring different states of the generator provides different levels (strong or weak) of observability. In \cite{ZhengITPWRS2021} a derivative free polynomial-chaos based method has been proposed to check the degree of observability of a NDAE model of power systems. Similarly in \cite{QiITPWRS2015} a method base on empirical observability Gramian has been proposed to assess the observability of synchronous generators using PMUs. 
To that end in this work however, we do not study the observability of NDAE power system models and we simply follow the literature. We assume that optimal number of PMUs are already placed in the network and all the states of the power system are completely observable.}

Based on the above discussion and aforementioned limitations in the surveyed studies, the objective of this paper is to produce a simple DSE algorithm that utilizes~\eqref{equ:PSModel}--\eqref{equ:PMUModel} to perform DSE for the complete NDAE representation of a power system under various sources of uncertainties. This contribution is the first in the literature of power systems DSE. The technical paper contributions are as follows: 
\begin{table}[t!]
	\textcolor{black}{
	\footnotesize	\renewcommand{\arraystretch}{1.3}
	\caption{Explanation of key notations used in the paper.}
	\label{tab:notation}
	\vspace{-0.1cm}
	\centering
	\begin{tabular}{|c|c|}
		\hline
		\textbf{Notation} & \textbf{Description}\\
		\hline
		\hline
		\hspace{-0.1cm}$\mathcal{E}$ and $\mathcal{N}$ & \hspace{-0.1cm} set of transmission lines and buses \\
		\hline
		\hspace{-0.1cm}$\mathcal{G}$, $\mathcal{L}$, and $\mathcal{R}$ & \hspace{-0.1cm} set of generator, load and renewable buses \\
		\hline
		\hspace{-0.1cm}$\delta_{i}$ and $\omega_{i}$ & \hspace{-0.1cm} generator rotor angle and frequency \\
		\hline
		\hspace{-0.1cm}$\omega_{0}$& \hspace{-0.1cm} synchronous speed ($2\pi60\;\mathrm{rad/sec}$) \\
		\hline
		\hspace{-0.1cm}$B_{ij}$ and $G_{ij}$& \hspace{-0.1cm} susceptance and conductance of line\\
		\hline
		\hspace{-0.1cm}$D_i$ & \hspace{-0.1cm} generator damping coefficient ($\mr{pu} \times \mr{sec}$)  \\
		\hline
		\hspace{-0.1cm}$E'_{di}$ and $E'_{qi}$ & \hspace{-0.1cm} generator transient voltage along \emph{dq}-axis ($\mathrm{pu}$) \\
		\hline
		\hspace{-0.1cm}$E_{\mr{fd}i}$ & \hspace{-0.1cm} generator field voltage  ($\mathrm{pu}$) \\
		\hline
		\hspace{-0.1cm}$\m I$ & \hspace{-0.1cm} identity matrix of appropriate dimension\\ 
		\hline
		\hspace{-0.1cm}$M_i$ & \hspace{-0.1cm} generator rotor inertia constant ($\mr{pu} \times \mr{sec}^2$) \\
		\hline
		\hspace{-0.1cm}$\m O$ & \hspace{-0.1cm} zero matrix of appropriate dimension\\
		\hline
		\hspace{-0.1cm}$P_{\mr{R}i},\;Q_{\mr{R}i}$ & \hspace{-0.1cm} active and reactive power from renewables ($\mr{pu}$)  \\
		\hline
		\hspace{-0.1cm}$P_{\mr{G}i},\;Q_{\mr{G}i}$ & \hspace{-0.1cm} active and reactive power from generators ($\mr{pu}$)  \\
		\hline
		\hspace{-0.1cm}$P_{\mr{L}i},\;Q_{\mr{L}i}$ & \hspace{-0.1cm} active and reactive load demand ($\mr{pu}$)  \\
		\hline
		\hspace{-0.1cm}$T'_{\mr{d0}i}$ and $T'_{\mr{q0}i}$ & \hspace{-0.1cm} open-circuit time constant along \emph{dq}-axis ($\mr{sec}$) \\
		\hline
		\hspace{-0.1cm}$T_{\mr{M}i}$ & \hspace{-0.1cm} generator mechanical input torque ($\mathrm{pu}$)\\
		\hline
		\hspace{-0.1cm}$\m u$ & \hspace{-0.1cm} system's  inputs vector \\ 
		\hline
		\hspace{-0.1cm}${v}_i$ and ${\theta}_i$& \hspace{-0.1cm} bus voltage and angle ($\mr{pu}$)  \\
		\hline
		\hspace{-0.1cm}$x_{\mr{d}i}$ and $x_{\mr{q}i}$& \hspace{-0.1cm} synchronous reactance along \emph{dq}-axis ($\mr{pu}$) \\
		\hline
		\hspace{-0.1cm}$x'_{\mr{d}i}$ and $x'_{\mr{q}i}$ & \hspace{-0.1cm} transient reactance along \emph{dq}-axis ($\mr{pu}$) \\
		\hline
		\hspace{-0.1cm}$\m x_{\mr{d}}$ and $\m x_{\mr{a}}$ & \hspace{-0.1cm} dynamic and algebraic variables  \\
		\hline
		\hspace{-0.1cm}$\m x$ and $\hat{\m x}$ & \hspace{-0.1cm} actual and estimated state vector \\ 
		\hline
		\hspace{-0.1cm}$\mathcal{X}_d$ and $\mathcal{X}_a$ & \hspace{-0.1cm} set containing upper and lower bounds of $\m x_d$ and $\m x_a$ \\ 
		\hline
		\hspace{-0.1cm}$\m y$ and $\hat{\m y}$& \hspace{-0.1cm} actual and estimated output vector \\ 
		\hline
		\hspace{-0.1cm}$\mathbb{R}^n$ & \hspace{-0.1cm} row vector of $n$ real numbers\\
		\hline
		\hspace{-0.1cm}$\mathbb{R}^{p\times q}$ & \hspace{-0.1cm} real matrix of size $p$-by-$q$\\
		\hline
		\hspace{-0.1cm}$\mathrm{Blkdiag}(\cdot)$ & \hspace{-0.1cm} generate a block diagonal matrix\\
		\hline
		\hspace{-0.1cm}$\mathrm{Diag}(\cdot)$ & \hspace{-0.1cm} generate a diagonal matrix\\
		\hline
		\hspace{-0.1cm}$\oslash$ and $\odot$ & \hspace{-0.1cm} Hadamard division and Hadamard product \\
		\hline
		\hspace{-0.1cm}$*$ & \hspace{-0.1cm} denotes symmetric entries in a symmetric matrix\\
		\hline
	\end{tabular}}
	\vspace{-0.3cm}
\end{table}

\begin{itemize}[leftmargin=*]
\item   
This work is first to propose DSE with NDAE representation of a power system having $(1)$ a higher order generator model with power balance equations of network and generator stator algebraic equations, $(2)$ linearization free DSE approach, $(3)$ more practical PMUs based measurement model, and $(4)$ simultaneous estimation of both algebraic and dynamic states of the power system.

\item To deal with the process and measurement noise and disturbances from load and renewables, we propose $H_{\infty}$ based NDAE observer which provides robust state estimation in the presence of Gaussian, non-Gaussian process and measurement noise, as well as uncertainty from loads and renewables. The main advantage of the $H_{\infty}$ NDAE observer over the stochastic estimators is it does not require any statistical properties of the disturbances. 
Albeit these types of observer designs are widely used in control theoretic literature \cite{ThabetITCST2018,ChenITAC2007,GuopingITCS2006,AlessandriITAC2020,PhamICTSL2019}, however, to the best of authors' knowledge no such work has been carried out to assess their applicability in power systems DSE. The proposed observer design is different from those provided in \cite{GuopingITCS2006,AlessandriITAC2020,PhamICTSL2019,9735348} as we have used $H_\infty$ and proportional integral (discussed below) notion and shaped the overall observer design as semidefinite optimization problem to synthesize a robust observer to perform DSE for power systems.  In particular, the formulated observer design is unique on its own as it includes constraints and objectives that allow for the observer to be practically implemented for power system \textbf{nonlinear} DAE model state estimation. 
\item To deal with the unknown control inputs we have extended the $H_{\infty}$ NDAE observer based on proportional integral (PI) framework. The main advantage of PI-based observer to handle unknown inputs over methods present in the literature is that there is no requirement for the generator buses to be equipped with PMUs \cite{ZhaoITPWRS2020}. The real-time implementation of the estimator is as simple as a one-step state predictor.
\end{itemize}

The rest of the paper is organized as follows. Section \ref{sec:DAE_ model} describes the nonlinear differential algebraic model of the power system. Section \ref{section:robust_obs_design} focuses on modeling uncertainty and the design of the robust observer under different sources of uncertainty.  Case studies are presented in Section \ref{section:simulations} and the paper is concluded in Section \ref{section:conclusion}.
\section{Nonlinear DAE Model Of Power Systems}\label{sec:DAE_ model}

We consider a graphical representation of a power system $(\mathcal{N},\mathcal{E})$, where $\mathcal{E} \subseteq \mathcal{N}\times\mathcal{N}$ are the total number of transmission lines, $\mathcal{N} = \mathcal{G} \cup \mathcal{L}$ are the total number of buses in the network while $\mathcal{G}$ and $\mathcal{L}$ are the set of generator and load buses respectively. The set of equations we used to describe our model are ordinary differential equations (ODEs), describing the generator dynamic model, and algebraic equations describing the power flow/balance equations. Combining both these equations we get the NDAE representation of the power system. A detailed description of these equations is given in the following section.
\vspace{-0.4cm}
\subsection{{Generator Dynamics and Algebraic Equations }}
We consider the standard two axes $4^{th}$-order transient model of synchronous generator $i$ $\in$ $\mathcal{G}$ which can be represented by the following differential equations \cite{sauer2017power}
\begin{subequations} \label{eq:SynGen}
	\begin{align}
	\dot{\delta}_{i} &= \omega_{i} - \omega_{0} \label{eq:SynGen1} \\ 
	\begin{split}
	M_{i}\dot{\omega}_{i} &= T_{\mr{M}i}-P_{\mr{G}i}- D_{i}(\omega_{i}-\omega_{0}) \end{split}\label{eq:SynGen2}    \\ 
	T'_{\mr{d0}i}\dot{E}'_{qi} &= -\tfrac{x_{\mr{d}i}}{x'_{\mr{d}i}}E'_{qi} +\tfrac{x_{\mr{d}i}-x'_{\mr{d}i}}{x'_{\mr{d}i}}v_i\cos(\delta_{i}-\theta_i) + E_{\mr{fd}i}  \label{eq:SynGen3} \\
T'_{\mr{q0}i}\dot{E}'_{di} &= -E'_{di} +\tfrac{x_{\mr{q}i}-x'_{\mr{q}i}}{x_{\mr{q}i}}v_i\sin(\delta_{i}-\theta_i)	 \label{eq:SynGen4}  
	\end{align} 
\end{subequations}
\textcolor{black}{where $\bmat{{\delta}_{i}&{\omega}_{i}&{E'}_{di}&{E'}_{qi}}$ are the four states of synchronous machine. The detailed explanation of each of these parameter is given in Table \ref{tab:notation}.}

\textcolor{black}{The algebraic constraints in the model are the power (active and reactive) flow equations and the model describing real and reactive power generated by the synchronous generators. These equations must be satisfied for all time instances and can be represented as  \cite{sauer2017power}}
\begin{subequations}\label{eq:SynGenPower}
	\begin{align}
		\begin{split}
			\hspace{-0.3cm}P_{\mr{G}i} &= \tfrac{1}{x'_{\mr{d}i}}E'_{qi}v_i\sin(\delta_i-\theta_i) -\tfrac{x_{\mr{q}i}-x'_{\mr{d}i}}{2x'_{\mr{d}i}x_{\mr{q}i}}v_i^2\sin(2(\delta_i-\theta_i))
		\end{split}
		\label{eq:SynGenPower1} \\
		\begin{split}
			\hspace{-0.3cm}Q_{\mr{G}i} &= \tfrac{1}{x'_{\mr{d}i}}E'_{qi}v_i\cos(\delta_i-\theta_i)-\tfrac{x'_{\mr{d}i}+x_{\mr{q}i}}{2x'_{\mr{d}i}x_{\mr{q}i}}v_i^2\\
			\hspace{-0.3cm}&\quad -\tfrac{x_{\mr{q}i}-x'_{\mr{d}i}}{2x'_{\mr{d}i}x_{\mr{q}i}}v_i^2\cos(2(\delta_i-\theta_i))
		\end{split}\label{eq:SynGePower2}
	\end{align}
\end{subequations}\\
\textcolor{black}{where  $i$ $\in$ $\mathcal{G}$. The power balance equation among generators, renewables and loads can be written as}
\begingroup
\allowdisplaybreaks 
\begin{subequations} \label{eq:GPF}
	\begin{align} 
		\begin{split}
			\hspace{-0.4cm}P_{\mr{G}i}+ P_{\mr{R}i} -P_{\mr{L}i} \hspace{-0.05cm}&=\hspace{-0.05cm} \sum_{j=1}^{N}\hspace{-0.05cm} v_iv_j\hspace{-0.05cm}\left(G_{ij}\cos \theta_{ij} \hspace{-0.05cm}+ \hspace{-0.05cm}B_{ij}\sin \theta_{ij}\right)
		\end{split}\label{eq:GPF1}\\
		\begin{split}
			\hspace{-0.4cm}Q_{\mr{G}i} + Q_{\mr{R}i}- Q_{\mr{L}i}\hspace{-0.05cm} &=\hspace{-0.05cm} \sum_{j=1}^{N}\hspace{-0.05cm} v_iv_j\hspace{-0.05cm}\left(G_{ij}\sin \theta_{ij} \hspace{-0.05cm}- \hspace{-0.05cm}B_{ij}\cos \theta_{ij}\right)
		\end{split}\label{eq:GPF2}		
	\end{align}
\end{subequations}
\endgroup
\textcolor{black}{where  $\theta_{ij}= \theta_i-\theta_j$ is the bus angle.  Similarly, for load buses  $i$ $\in$ $\mathcal{L}$ the power flow equations can be written in a same fashion as \eqref{eq:GPF} with the exception that $P_{Gi}= Q_{Gi} = P_{Ri}= Q_{Ri} = 0$. Note that in this paper we are modeling renewables as a negative load meaning they are injecting power into the network.}

To proceed, we define ${\m x}_d\hspace{-0.2cm}=\hspace{-0.3cm}\bmat{\m \delta^\top\;\m \omega^\top\;\m E'^\top_{q}\; {{\mE'}}^\top_{d}}^\top$ as the dynamic states,  ${\m x}_a\hspace{-0.1cm} = \hspace{-0.1cm}\bmat{\m P^\top_{G} &\m Q^\top_{G} &\m v^\top &\m \theta^\top}^\top$ as the algebraic variables, $\m q\hspace{-0.2cm} = \hspace{-0.2cm}\bmat{\m P_{R}^{\top}&\m Q_{R}^{\top}&\m P_{L}^{\top}&\m Q_{L}^{\top}}^\top$, and $\m u \hspace{-0.2cm}=\hspace{-0.2cm} [\m T_{\mathrm{M}}^{\top}\;\m E_{\mathrm{fd}}^{\top}]^\top$, where  ${\m P_{G}}\hspace{-0.05cm}=\hspace{-0.05cm}\{P_{Gi}\}_{i\in \mc{G}}\hspace{-0.05cm}$,   ${\m Q_{G}}\hspace{-0.05cm}=\hspace{-0.05cm}\{Q_{Gi}\}_{i\in \mc{G}}\hspace{-0.05cm}$, ${\m P_{L}}\hspace{-0.05cm}=\hspace{-0.05cm}\{P_{Li}\}_{i\in \mc{L}}\hspace{-0.05cm}$,   ${\m Q_{L}}\hspace{-0.05cm}=\hspace{-0.05cm}\{Q_{Li}\}_{i\in \mc{L}}\hspace{-0.05cm}$, 
${\m v}\hspace{-0.05cm}=\hspace{-0.05cm}\{v_i\}_{i\in \mc{N}}\hspace{-0.05cm}$, ${\m\theta}\hspace{-0.05cm}=\hspace{-0.05cm}\{\theta_i\}_{i\in \mc{N}}\hspace{-0.05cm}$, ${\m \delta}\hspace{-0.05cm}=\hspace{-0.05cm}\{\delta_i\}_{i\in \mc{G}}\hspace{-0.05cm}$, ${\m \omega}\hspace{-0.05cm}=\hspace{-0.05cm}\{\omega_i\}_{i\in \mc{G}}\hspace{-0.05cm}$,  ${\m E'_{q}}\hspace{-0.05cm}=\hspace{-0.05cm}\{E'_{qi}\}_{i\in \mc{G}}\hspace{-0.05cm}$, and  ${\m E'_{d}}\hspace{-0.05cm}=\hspace{-0.05cm}\{E'_{di}\}_{i\in \mc{G}}\hspace{-0.05cm}$. 
Based on the above vectors descriptions, the 
NDAE model \eqref{eq:SynGen}--\eqref{eq:GPF} of a power system can be represented as
\begin{subequations}\label{eq:nonlinearDAE}
	\begin{align}
		\dot{{\m x}}_d &= {\m A}_d{\m x}_d +  {\m F}_d{\m f}_d\left({\m x}_d,{\m x}_a\right) + {\m B}_d {\m u} + {\m h} \omega_{0} \label{eq:nonlinearDAE-1}\\
		\m 0 &= {\m A}_a{\m x_a} + {\m F}_a{\m f}_a\left({\m x}_d,{\m x}_a\right) + {\m B}_a {\m q}\label{eq:nonlinearDAE-2}
	\end{align}
\end{subequations}
where  $\m x_a \in \mbb{R}^{n_a}$, $\m x_d \in \mbb{R}^{n_d}$, $\m u \in \mbb{R}^{n_u}$, and $\m q \in \mbb{R}^{n_q}$. The functions $\m f_a:\mathbb{R}^{n_d}\times \mathbb{R}^{n_a}\rightarrow \mathbb{R}^{n_{fa}}$, and $\m f_d:\mathbb{R}^{n_d}\times \mathbb{R}^{n_a}\rightarrow \mathbb{R}^{n_{fd}}$ describes the nonlinearities in the algebraic and dynamic states respectively, while the rest of the constant matrices ${\m A}_a\in \mbb{R}^{n_a\times n_a}$, ${\m F}_a\in \mbb{R}^{n_{a}\times n_{fa} }$,  ${\m B}_a\in \mbb{R}^{n_{a}\times  n_q}$, ${\m A}_d\in \mbb{R}^{n_d\times n_d}$, ${\m F}_d\in \mbb{R}^{n_{d}\times n_{fd} }$ , ${\m B}_d\in \mbb{R}^{ n_{d}\times  n_u}$, and $\m h\in \mbb{R}^{n_d}$ are all detailed in Appendix \ref{appdx:A}. Considering an overall state vector $ \m {x} = \bmat{\m x_d^\top & \m x_a^\top}^\top \in \mbb{R}^n$ then the model detailed in \eqref{eq:nonlinearDAE} can be rewritten in a compact form as 
\begin{align}\label{eq:nonlinearDAEexplicit-1}
		\begin{split}
		{{\m Z}}\dot{{\m x}} &= {{\m A}}\m{x} +  {{\m F}}{\m f}\left( \m{x}\right) + \m B_u \m u + \m B_q \m q + \m H \omega_{0}
	\end{split}
	\end{align}
\textcolor{black}{where
\begin{align*}
\begin{split}
\m Z\hspace{-0.1cm}=\hspace{-0.1cm}\bmat{\m I& \mO\\\mO&\mO},{\m A}\hspace{-0.1cm}=\hspace{-0.1cm}\bmat{\m A_d& \mO\\ \mO&\m A_a},\m F\hspace{-0.1cm}=\hspace{-0.1cm}\bmat{\m F_d& \mO\\\mO&\m F_a},\m H\hspace{-0.1cm}=\hspace{-0.1cm}\bmat{\m h\\ \m O}
\end{split}\\
\begin{split}
{\m f}\left( x\right)= \bmat{{\m f}_d\left(\m x\right)\\{\m f}_a\left(\m x\right)},\, \mB_u = \bmat{\m B_d& \mO}^\top, \mB_q = \bmat{\mO&\m B_a}^\top
\end{split}
\end{align*}}
In the following section we bound the nonlinearities using Lipschitz continuity condition and present the PMU based measurement model.
\subsection{Bounding Nonlinearities and PMU Measurement Model}

The vector $\m f(.)$ encapsulates the nonlinearities in dynamic $\m f_d(.)$ and algebraic equations $\m f_a(.)$. In this work we assume that the function $\m f(.)$ is Lipschitz continuous, which means $\m f(.)$ is continuously differentiable and the magnitude of the derivative is bounded above by a constant real number. 

We note that this assumption holds herein as there are indeed upper and lower bounds on the states (e.g., frequency, voltages) of the system \cite{SebastianITPWRS2020}. To that end, let us consider two sets $\mathcal{\mX}_d$ and $\mathcal{\m X}_a$ as follows
\begin{subequations}
	\begin{align}
		\mathcal{\mX}_d &= [\underline{\delta}, \overline \delta]\times [\underline{\omega}, \overline \omega]\times [\underline{E}'_q, \overline { E}'_q]\times[\underline{E}'_d, \overline { E}'_d] \label{Lipschitz1} \\
		\mathcal{\mX}_a &= [\underline{v}, \overline v]\times [\underline{\theta}, \overline \theta] \label{Lipschitz2}
	\end{align}
\end{subequations}
where $\m x_d \in \mathcal{\m X}_d$ and $\m x_a \in \mathcal{\m X}_a$. Eqs. \eqref{Lipschitz1} and \eqref{Lipschitz2} define upper and lower bounds or operating regions of the state variables of NDAE model \eqref{eq:nonlinearDAEexplicit-1}. These operating regions can be determined based on the operator's knowledge or by performing extensive simulations studies while applying different contingencies and then finding the operating regions of the generators \cite{QiIETPWRS2017}. The Lipschitz continuity condition can be written as
\begin{align}\label{eq:lipshitz}
	\norm{\m f(\m x) - \m f(\hat{\m x})}_2 \leq \norm {\m G(\m x-\hat{\m x})}_2
\end{align}
where $\m G  \in \mbb{R}^{n\times n}$ is a constant diagonal matrix and it contain Lipschitz constants for each of the corresponding nonlinearity in algebraic and dynamic equation of NDAE model \eqref{eq:nonlinearDAEexplicit-1}. These Lipschitz constants  can be determined numerically given that  $\mathcal{\m X}_d$ and $\mathcal{\m X}_a$ are known. Readers are referred to \cite{SebastianACC2019} for the complete method explaining the computation of these Lipschitz constants.

\textcolor{black}{As for the PMU measurement model, let us consider a vector $\m y = \bmat{\m{\mr V}^\top&\m{\mr I}^\top}^\top\in\mbb{R}^p$ denoting the measurements received from PMUs, where  ${\m{\mr V}} = \bmat{\{v_{Rj}\}_{j\in \mathcal N} +\{v_{Ij}\}_{j\in \mathcal N}}$ represents voltage phasors and $\m {\mr I} = \bmat{\{I_{Rji}\}_{i\in \mathcal N_j} +\{I_{Iji}\}_{i\in \mathcal N_j}}$ denotes current phasors of bus $j$ and nearby buses $(\mathcal N_j \in \mc N)$ connected with bus $j$. Then the NDAE model of a multi-machine system with PMUs measurement model can be expressed as}
\begin{subequations}\label{eq:final_DAE}
	\begin{align}
		\m Z\dot{{\m x}} &= \m A\m x +  {\m F}{\m f}\left(\m x\right) + \m B_u \m u + \m B_q \m q + {\m H} \omega_{0}+ \m w_p \label{eq:final_DAE1}\\
		\m y &= \m C\m x + \m w_m \label{eq:final_DAE2}
	\end{align}
\end{subequations}
\textcolor{black}{where $\m w_m \in\mbb{R}^p$ denotes measurements noise, $\m w_p \in\mbb{R}^p$ represent random process noise, and $\mC\in\mbb{R}^{p\times n}$ is a constant output matrix that maps state vector $\m x$ to what typically PMUs measure (i.e., voltage and current phasors). The overall structure of $\mC$ is detailed in Appendix \ref{Appndx:pmu measurement}.}
\vspace{-0.2cm}
\section{Joint Estimator for NDAE States}\label{section:robust_obs_design}

\textcolor{black}{The NDAE model presented in \eqref{eq:final_DAE} assumes ideal power system conditions. However, there are always different types of disturbances and unknown inputs that ought to be modeled when designing a state estimation method. With that in mind, to model disturbances from load and renewables we consider that minutes- or hour-ahead predictions of these quantities are available but the disturbances are unknown.  Accordingly, one can write $\m q(t) = \bar{\m q} + \Delta \m q(t)$ where $\bar{\m q}$ is the known part and $\Delta \m q(t)$ is the unknown uncertainty. Similarly for generator's control inputs we can write  $\m u(t) = \bar{\m u} + \Delta \m u(t)$ where $\bar{\m u}$ is the known steady state value of the inputs and $\Delta \m u(t)$ model the disturbances. With that in mind the NDAE model  \eqref{eq:final_DAE} can be rewritten as
\begin{subequations}\label{eq:NDAE_with_noise}
	\begin{align}
	\hspace{-0.5cm}	\m Z\dot{{\m x}} &= \m A\m x +  \m F\m f\left(\m x\right) + { {\m {B}}_u {{\m {u}} } + \m B_q \bar{\m q} + {\m H} \omega_{0}} + {\m {B}_w} {{\m w} }  \label{eq:NDAE_with_noise_a}\\
		\m y &= \m C\m x  + {\m {D}_w} {{\m w}. }\label{eq:NDAE_with_noise_b}
	\end{align}
\end{subequations}
The matrices ${\m {B}_w}$ and ${\m {D}_w}$ are constant known matrices and they maps  $\m w$ into the system's dynamics and PMUs measurements. 
These matrices are constructed as ${\m {B}_w}\hspace{-0.2cm} =\hspace{-0.2cm} \Blkdiag\left( \mI,\, \mO,\, \mB_{w_R},\, \mB_{w_L}\right)\hspace{-0.2cm} \in\hspace{-0.2cm} \mbb{R}^{n\times q}$ and $\m{D}_w \hspace{-0.2cm}=\hspace{-0.2cm} \Blkdiag\left( \mO,\, \mI,\, \mO,\, \mO\right) \hspace{-0.06cm}\in\hspace{-0.07cm}\mbb{R}^{p\times q}$, where $\mB_{w_R} \in \mbb{R}^{\mc{N}\times\mc{N}}$ is a binary matrix and has $1's$ at those locations where renewables are connected to buses and \textit{zero} otherwise, similarly $\mB_{w_L} \in \mbb{R}^{\mc{N}\times\mc{N}}$ has  $1's$ only at those locations where buses are connected to loads.} 
 
\vspace{-0.2cm}
\subsection{$H_\infty$ Stability and Observer Design}\label{sec:Robust_sec_A}
\vspace{-0.1cm}
To begin with the observer design, let $\hat{\m x}$ be the estimated states and $\hat{\m y}$ be the estimated outputs, then the proposed estimator/observer dynamics for the NDAE model \eqref{eq:NDAE_with_noise} can be written as
\begin{subequations} \label{eq:obsr_dynamics}
	\begin{align} 
		\begin{split}
			\m Z\dot{\hat{\m x}}\hspace{-0.05cm} &=\hspace{-0.05cm} {\m A}{\hat{\m x}}\hspace{-0.05cm}\hspace{-0.05cm} + \hspace{-0.05cm} {\m F}{\m f}\left(\hat {\m x}\right)\hspace{-0.05cm} +\hspace{-0.05cm}\m L\left( \m y - \hat{\m y}\right)\hspace{-0.05cm} + \hspace{-0.05cm} {\m {B}}_u {{\m {u}} } \hspace{-0.05cm}+\hspace{-0.05cm} \m B_q \bar{\m q} \hspace{-0.05cm}+\hspace{-0.05cm} {\m H} \omega_{0}
		\end{split} \label{eq:obsr_dynamics1}\\
		\begin{split}
			\hat{\m y} &= \m C \hat{\m x}
		\end{split} \label{eq:obsr_dynamics2}
	\end{align}
\end{subequations}
where $\mL \in \mbb{R}^{n\times p}$ is the Luenberger type gain matrix. Let $\m e = \m x - \hat{\m x}$ be the error between estimated and actual states. Multiplying $\m Z$ on both sides and taking the derivative, then  the error dynamics can be written as
\begin{align}\label{eq:errorDyn}
	\m Z\dot{{\m e}} = \m Z\dot{{\m x}} - \m Z\dot { \hat{\m x}}.
\end{align}
Now putting values of $\m Z\dot{{\m x}}$ and $\m Z\dot { \hat{\m x}}$ from \eqref{eq:NDAE_with_noise_a}--\eqref{eq:obsr_dynamics1} and simplifying, the estimation error dynamics \eqref{eq:errorDyn} can be rewritten as
\begin{align}\label{eq:errorDyn_Hinf}
	\m Z\dot{{\m e}} &= (\m A -\m L \m C)\m e + \m F \Delta \m f+ (\mB_w - \mL \mD_w)\m w
\end{align}
where $ \Delta \m f =\m f \left(\m x\right)-\m f \left(\hat{\m x}\right)$. Our main objective throughout this paper is to design observer gain $\mL$ such that the estimation error dynamics \eqref{eq:errorDyn_Hinf} converges asymptotically to zero and robust performance from the observer can be achieved under various sources of unknown disturbances

With that in mind, we introduce the $H_\infty$ stability notion to minimize the impact of disturbance vector $\m w$ on the estimation error dynamics. In state estimation theory the $H_\infty$ notion was first introduced in \cite{ShakedITAC1990} to design an optimal state estimator for a linear system subject to disturbances. The main advantage of $H_\infty$ norm minimization over Kalman filters is that it does not require any prior knowledge about the statistic of the noise. In $H_\infty$ based state estimation, the noise/disturbances are considered as arbitrary bounded signals and the observer is designed to ensure a specified $H_\infty$ performance for the error dynamics for all disturbances.

To that end, the $H_\infty$ stability definition can be applied to estimation error dynamics \eqref{eq:errorDyn_Hinf} in the following fashion.
\begin{mydef}\label{def:H_inf}
Let $\m \beta = \m\varGamma \m e$ be the performance of error dynamics with $\m\varGamma \in \mbb{R}^{n\times n}$ as the user defined performance matrix. Then the nonlinear estimation error dynamics \eqref{eq:errorDyn_Hinf} is $H_{\infty}$ stable with performance level  $\gamma$ if, $(a)$ \eqref{eq:errorDyn_Hinf} is stable when $\m w = 0$ for all $t>0$, $(b)$ $\norm{{\m\beta }}^2_{L_2} <  \gamma\norm{{\m w }}^2_{L_2}$ for zero initial error ($\m e = 0$) and for any bounded disturbances $\m w$.
\end{mydef}
Definition \ref{def:H_inf} with performance level $\gamma$ can be interpreted as follows. When the initial value of the error is zero i.e., $\m e = 0$, then for $t>0$ the magnitude of the performance vector $\m \beta$ is guaranteed to evolve in a way such that its value is always less than the constant time the magnitude of disturbance $\m w$. Note that if the value of $\gamma$ is large then it means that the performance of error dynamics $\m \beta$ can highly be affected by the disturbance vector $\m w$. Thus in the $H_{\infty}$ based observer design we want to minimize $\gamma$ as small as possible to get robust performance of the observer under various sources of disturbances. 

To that end, we now propose a systematic way based on Lyapunov stability criteria to synthesize $H_{\infty}$ based observer for the NDAE described in \eqref{eq:NDAE_with_noise}.
\begin{theorem}\label{theorm:H_inf}
Consider the NDAE model \eqref{eq:NDAE_with_noise} and observer dynamics \eqref{eq:obsr_dynamics}. Then the estimation error dynamics \eqref{eq:errorDyn_Hinf}  converges asymptotically to zero and is robust in $H_\infty$ sense against disturbances, if there exists positive definite matrix $\m X \in \mbb{R}^{n \times n}$, two matrices $\mR\in \mbb{R}^{p\times n}$ and $\mY\in \mbb{R}^{n_a\times n}$ and a scalar ${\epsilon} \in\mbb{R}_{++}$ and $\gamma > 0$ such that the following semidefinite optimization problem is solved.
\begin{align*}
\mathbf{\left( P_1\right) }\;\minimize_{{\epsilon},\gamma, \m X,\m R, \m Y} &\;\;\; c_1\gamma\\ \subjectto  & \;\;\;\mr{LMI}\; \eqref{eq:LMI_Hinf},\;\m X \succ 0,\;{\epsilon} > 0,\gamma>0
\end{align*} 
where $\mr{LMI}$ \eqref{eq:LMI_Hinf} is given as
	\begin{align}\label{eq:LMI_Hinf}
		\hspace{-0.1cm}\bmat{ \m\Omega & * & * \\(\mX\mZ+\mZ^{\perp\top}\mY)^\top \m F & -{\epsilon}\m I &* \\
			\m{B_w}^\top(\mX\mZ+\mZ^{\perp\top}\mY)-\m{D_w}^\top\m{R}&\m O&-{\gamma}\m I} \prec 0 
	\end{align}
and $\m\Omega $ is as follows
	\begin{align*}
		\begin{split}
			\m\Omega = \mA^\top(\mX\mZ+\mZ^{\perp\top}\mY) + (\mX\mZ+\mZ^{\perp\top}\mY)^\top\mA -\\
			\m C^\top \m R-\m R^\top \m C +\epsilon\m G^\top \m G+\m \varGamma ^\top \m \varGamma 
		\end{split}
	\end{align*}
 where $\mZ^{\perp\top}\in\mbb{R}^{n_a\times n}$ is the orthogonal  complement of $\mZ$, meaning $\mZ^{\perp\top}\mZ = 0$. Upon solving $\mathbf{P_1}$ the observer gain can be recovered as $\m L = \m P^{-\top} \m R^\top$
\end{theorem}
\begin{proof}\textcolor{black}{
According to \cite[Theorem 2.1]{GuopingITCS2006} there exists an observer for a nonlinear differential algebraic model with nonlinearities as Lipschitz bounded (same as the one shown in Eq. \eqref{eq:final_DAE}) if there exists  two matrices $\m P\in\mbb{R}^{n\times n}$ and $\mQ\in \mbb{R}^{p\times n}$ such that the following matrix inequalities are solved 
\begin{subequations}\label{eq:LMI_core}
	\begin{align}
		\bmat{\m \Omega& \mP^\top \m B \\ \mB^\top \m P & \m {-I}} \prec 0 \label{eq:LMI_core_a} \\
		\m Z^\top \m P =\m P^\top \m Z \succeq 0 \label{eq:LMI_core_b}
	\end{align}
\end{subequations}
where  $\m \Omega = \m A^\top \m P+\m P^\top \m A+\m C^\top \m Q + \m Q^\top \m C + \m F^\top \m F$ and observer gain $\m L$ can be recovered as $\m L = \m P^{-\top} \m Q^\top$.  In a similar fashion we can design our observer for the NDAE model presented in \eqref{eq:NDAE_with_noise}. 
The overall proof is divided into two main steps as follows:
\begin{enumerate}[label=(\alph*)]
	\item Determining matrix inequalities based observer design.
	\item Converting the matrix inequalities to linear matrix inequalities (LMIs) so that they can be easily solved using commercially available SDP optimization solvers such as MOSEK \cite{Andersen2000}.
\end{enumerate}
\textcolor{black}{
	\textbf{\textit{(a):}} Let us consider a candidate Lyapunov function as $V(\m e)=\m e^\top \mZ^\top \mP \m e $, where $\m P\in\mbb{R}^{n\times n}$, $V:\mbb{R}^{n}\rightarrow \mbb{R}_+$, $\mZ^\top \mP = \mP^\top \mZ \succeq 0$, then its derivative can be written as
	\begin{align*}
		{\dot V}(\m e) &= (\m{Z}\dot{\m e})^\top \mP \m e+(\m{Z}{\m e})^\top (\mP \dot{\m e}).
	\end{align*}
	Since $\mZ^\top \mP = \mP^\top \mZ $, then we can write
	\begin{align*}
		{\dot V}(\m e) &= (\m{Z}\dot{\m e})^\top \mP \m e+(\mP \m e)^\top(\m{Z}\dot{\m e}).
\end{align*}}
Putting value of $\m{Z}\dot{\m e}$ from equation \eqref{eq:errorDyn_Hinf} yields
\begin{align*}
	{\dot V}(\m e) &= \left( \m{A}_c\m e +{\m F\Delta \m f( \m x,{\hat {\m x}})} + (\m{B}_w-\m{LD}_w)\m w\right)^\top\m{P e} +\\
	&{(\m P\m e)}^\top\left( \m{A}_c\m e +{\m F\Delta \m f( \m x,{\hat {\m x}})} + (\m{B}_w-\m{LD}_w)\m w\right)
\end{align*}
where $ \m{A}_c = (\m A-\mL\mC)$. 
Now for any bounded disturbances $\m w$ the $H_{\infty}$ stability condition is; ${\dot V}(\m e) + \m \beta^\top\m\beta - \gamma\m w^\top\m w < 0$ thus 
\begin{align*}
	(\m{A}_c\m e +&{\m F\Delta \m f( \m x,{\hat {\m x}})} + (\m{B}_w-\m{LD}_w)\m w)^\top\m{P e} +\\
	&{(\m P\m e)}^\top\left( \m{A}_c\m e +{\m F\Delta \m f( \m x,{\hat {\m x}})} + (\m{B}_w-\m{LD}_w)\m w\right) +\\
	& \hspace{4cm} \m \beta^\top\m\beta - \gamma\m w^\top\m w < 0.
\end{align*}
These equation can be rearranged and written as $\m\varPsi^\top\m\Theta\m\varPsi<0$. where
\begin{align*}
	\m\varPsi = \bmat{\m e\\\Delta \m f\\\m w}, \m\Theta=\hspace{-0.1cm}\bmat{ \m\Theta_{11} & * & * \\\mF^\top \m P & \m O &* \\
		(\m{B}_w-\m{LD}_w)^\top\m{P}&\m O&-{\gamma}\m I}
\end{align*}
and $ \m\Theta_{11}$ is given as
\begin{align*}
	\begin{split}
		\m\Theta_{11} = (\m A-\mL\mC)^\top \m P+\m P^\top  (\m A-\mL\mC)
		+\m \varGamma ^\top \m \varGamma.
	\end{split}
\end{align*}
Note that $\m\varPsi^\top\m\Theta\m\varPsi<0$ holds only if $\m\Theta\prec0$. Now from Eq. \eqref{eq:lipshitz} we know that the function $\m f(.)$ is Lipschitz bounded, meaning
\begin{align*}
	\norm{\Delta  {\m f(\m x,{\hat {\m x}})}}_2 \leq \norm {{\m G( \m x-{\hat{\m x}})}}_2\\
	\Leftrightarrow (\Delta {\m f( \m x,{\hat {\m x}})}^\top \Delta  {\m f(\m x,{\hat {\m x}})})- \m e^\top\mG^\top \mG \m e\leq 0
\end{align*}
which can be written as $\m\varPsi^\top\m \Xi\m\varPsi\leq0$, where
\begin{align*}
	\m \Xi = \mr{diag}\left( \bmat{-\mG^\top\mG&\mI&\mO}\right) 
\end{align*}
since $\m\varPsi^\top\m \Xi\m\varPsi\leq0$ for all admissible $\m\varPsi$ then it means $\m \Xi\prec0$. From S-Lemma \cite{Slemma}, $\m\Theta\prec0$ if there exists $\epsilon \geq 0 $ such that  $\m\Theta-(\epsilon)\m \Xi \prec0$. Thus the total matrix inequalities that we need to solve for observer design can be written as
\begin{subequations}\label{eq:matrix_ineq}
	\begin{align}
		\bmat{\m\Upsilon & * & * \\\mF^\top \m P & -\epsilon\mI &*\\
			(\m{B}_w-\m{LD}_w)^\top\m{P}&\m O&-{\gamma}\m I} &\prec 0 \label{eq:matrix_ineq1}\\
		\m Z^\top \m P=\m P^\top \m Z &\succeq 0 \label{eq:matrix_ineq2}
	\end{align}
\end{subequations}
where $\m\Upsilon$ is given as
\begin{align*}
	\begin{split}
		\m\Upsilon = (\m A-\mL\mC)^\top \m P+\m P^\top  (\m A-\mL\mC) + \epsilon\mG^\top\mG
		+\m \varGamma ^\top \m \varGamma.
	\end{split}
\end{align*}
\textit{\textbf{(b):}} Now to make \eqref{eq:matrix_ineq} strict linear matrix inequality, we need to: \textit{{(1)}} eliminate the product of $\mL$ and $\mP$ from \eqref{eq:matrix_ineq1} (because both $\mL$ and $\mP$ are variables and there product make the matrix inequality nonlinear),  and \textit{{(2)}} eliminate  \eqref{eq:matrix_ineq2} (because it has equality terms). To tackle \textit{{(1)}}, lets assume $\mR = \mL^\top\mP\in\mbb{R}^{p\times n}$, then the product of $\mL$ and $\mP$ in \eqref{eq:matrix_ineq1} can be replaced by $\mR$.
To deal with \textit{{(2)}}, let us assume there exists two matrices ${\m M}\in \mbb{R}^{n\times n}$ and ${\m N}\in \mbb{R}^{n\times n}$ such that 
\begin{align}\label{eq:proof-eq-1}
	\begin{split}
		\m{MZN} &= \bmat{\m I & \m O\\ \mO & \mO}, 
	\end{split}
	\begin{split}
		\m{M^{-\top} PN} &= \bmat{\m {P_1} & \m {P_2}\\ \m{P_3} & \m{P_4}} 
	\end{split}
\end{align}
where $\m {P_1} \in  \mbb{R}^{n_d\times n_d}$, $\m {P_2} \in  \mbb{R}^{n_d\times n_a}$, $\m {P_3} \in  \mbb{R}^{n_a\times n_d}$ and $\m {P_4} \in  \mbb{R}^{n_a\times n_a}$. Now from \eqref{eq:proof-eq-1} we can get 
\begin{subequations}\label{eq:proof-eq-2}
	\begin{align}
		\begin{split}\label{eq:proof-eq-2a}
			\m{Z} &= \mM^{-1}\bmat{\m I & \m O\\ \mO & \mO}\mN^{-1} 
		\end{split}\\
		\begin{split}\label{eq:proof-eq-2b}
			\m{P} &= \mM^\top\bmat{\m {P_1} & \m {P_2}\\ \m{P_3} & \m{P_4}}\mN^{-1}. 
		\end{split}
	\end{align}
\end{subequations}
From \eqref{eq:proof-eq-2a} and \eqref{eq:proof-eq-2b}, $\mZ^\top \mP$ and $\mP^\top \mZ$ is equal to
\begin{subequations}\label{eq:proof-eq-3}
	\begin{align}
		\begin{split}\label{eq:proof-eq-3a}
			\mZ^\top \mP = \mN^{-\top}\bmat{\m{P_1}&\mO\\\mO&\mO}\mN^{-1}
		\end{split}\\
		\begin{split}\label{eq:proof-eq-3b}
			\mP^\top \mZ = \mN^{-\top}\bmat{\m{P_1}^\top&\mO\\\m{P_2}&\mO}\mN^{-1}.
		\end{split}
	\end{align}
\end{subequations}
Now \eqref{eq:proof-eq-3a} and \eqref{eq:proof-eq-3b} can be made equal if $\m{P_1} = \m{P_1}^\top$ and $\m{P_2} = 0$. Hence, \eqref{eq:proof-eq-2b} can be written as
\begin{align*}
	\begin{split}
		\m{P} &\hspace{-0.0cm}= \hspace{-0.0cm}\mM^\top\bmat{\m {P_1} & \m {O}\\ \m{P_3} & \m{P_4}}\mN^{-1}\\
		&\hspace{-0.0cm}= \hspace{-0.0cm}\mM^\top\left( \bmat{\m {P_1} & \m {O}\\ \m{O} & \m{O}}+ \bmat{\m {O} & \m {O}\\ \m{P_3} & \m{P_4}}\right) \mN^{-1}
		\end{split}
	\end{align*}
\begin{align*}
		\begin{split}
		&\hspace{-0.0cm}= \hspace{-0.1cm}\mM^\top\left( \bmat{\m {P_1}\,\, \m {O}\\ \m{O} \,\,\,\, \m{I}}\bmat{\m {I} \,\,\,\, \m {O}\\ \m{O} \,\, \m{O}}\right) \mN^{-1}\hspace{-0.1cm} + \hspace{-0.1cm}\mM^\top\bmat{\m {O} \,\,\,\,\, \m {O}\\ \m{P_3} \,\, \m{P_4}}\mN^{-1}.
	\end{split}
\end{align*}
By defining $\mX$ as 
\begin{align*}
	\m{X} &= \mM^\top\bmat{\m {P_1} & \m {O}\\ \m{O} & \m{I}}\mM
\end{align*}
we can write the above equation for $\mP$ as 
\begin{subequations}\label{eq:proof-eq-final}
	\begin{align}
		\begin{split}
			\m{P} &= \m{XZ} + \underbrace{\mM^\top\bmat{\mO\\\mI}}_{\mZ^{\perp\top}}\underbrace{\bmat{\m{P_3}&\m{P_4}}\mN^{-1}}_{\mY}
		\end{split}\label{eq:proof-eq-finala}\\
		\begin{split}
			\m{P} &= \m{XZ}+{\mZ^{\perp\top}}{\mY}\label{eq:proof-eq-finalb}
		\end{split}
	\end{align}
\end{subequations}
\textcolor{black}{where $\mZ^{\perp\top}\in\mbb{R}^{n\times n}$ is the orthogonal  complement of $\mZ$ and $\mY\in \mbb{R}^{n_a\times n}$.  Hence, by plugging the value of $\mP$ from \eqref{eq:proof-eq-finalb} into \eqref{eq:matrix_ineq1} we can eliminate Eq. \eqref{eq:matrix_ineq2} which yields the strict LMI as shown in \eqref{eq:LMI_Hinf}}. This end the proof.}
\vspace{-0.2cm}
\end{proof}
In Theorem \ref{theorm:H_inf} we posed the calculation of observer gain  as a linear semidefinite optimization problem and thus $\mP_1$ can be solved using various off-the-shelf optimization solvers such as MOSEK \cite{LofbergICRA2004}. The calculated
observer gain $\mL$ ensures that the performance of error dynamics $\norm{{\m\beta }}^2_{L_2}$ or $\norm{{\m e }}^2_{L_2}$ 
is robust in $H_\infty$ sense as according to Definition \ref{def:H_inf}, or in other words it makes sure that $\norm{{\m\beta }}^2_{L_2}$ always lies within a circle of radius  $\gamma\norm{{\m w }}^2_{L_2}$ and origin at zero. Theorem \ref{theorm:H_inf} also guarantees that the estimation error dynamics \eqref{eq:errorDyn_Hinf} converges asymptotically to zero as $t\rightarrow\infty$. 

However, note that in optimization problem $\mathbf{P_1}$ there is no constraint on the magnitude of the observer gain $\mL$ and for practical reasoning, high gain observers are generally undesirable because they can increase the sensitivity of a system to disturbances. Moreover, albeit solving $\mathbf{P_1}$ ensures the stability of error dynamics, the convergence rate at which $\m {\hat x}$ approaches $\m x$ can be relatively poor. To that end, the optimization problem $\mathbf{P_1}$ can be improved with the following objectives: $(1)$ minimizing the maximum eigenvalues of $\mZ^\top\mX\mZ$. This is because in the proof of Theorem \ref{theorm:H_inf} we have assumed the Lyapunov candidate function to be $\m V = \m e^\top\m Z^\top\mP \m e^\top$ and thus to increase the convergence rate one can minimize the maximum eigenvalues of $\m Z^\top\mP$, and as $\m{P} = \m{XZ}+{\mZ^{\perp\top}}{\mY}$, thus, $\m Z^\top\mP = $ $\m Z^\top(\mX\mZ+\mZ^{\perp\top}\mY)=\mZ^\top\mX\mZ$, and $(2)$ minimizing the norm of matrix $\mR$ to get observer gain matrix $\mL$ of reasonable magnitude. Thus, the overall optimization problem we seek to solve can be written as
\begin{align*}
	\mathbf{\left( P_2\right) }\;\minimize_{{\epsilon},\kappa,\gamma, \m X,\m R, \m Y}  \;\;\;& c_1\kappa+c_2\gamma+c_3\norm{\m R}_2\\ \subjectto \;\;\; & \mr{LMI}\; \eqref{eq:LMI_Hinf},\;\m X \succ 0,\;{\epsilon}> 0,\gamma > 0,\kappa>0,\\ & \kappa\mI-\mZ^\top\mX\mZ\succ 0
\end{align*}
where $c_1, c_2$ and $c_3$ are predefined weighting constants and their values can be adjusted based on the specific requirements. For example, if high convergence is desired as compared to other variables then the value of $c_1$ can be increased, vice versa.

It is worth mentioning that the observer proposed in this study is different from those proposed in \cite{GuopingITCS2006} and \cite{PhamICTSL2019}. As compared to \cite{GuopingITCS2006} we are minimizing $\kappa$,  $\gamma$, and $\mR$ which ensures, quick convergence, robust performance and observer gain of reasonable magnitude. Similarly, in \cite{PhamICTSL2019} although $H_\infty$ stability notion is used to handle disturbances, the presented observer is not able to handle unknown control inputs, while in this study proportional integral notion (as discussed in the following section) is used which provide robust estimation in the presence of unknown control inputs. 
\vspace{-0.2cm}
\subsection{Tackling Unknown Control Inputs}\label{section:robust_B}
\vspace{-0.1cm}
The observer designed in Section~\ref{sec:Robust_sec_A} deals with uncertainty from process and measurement noise as well as renewables and loads. Herein, we deal with the deviations associated with generator's control inputs. To that end, the NDAE dynamics can be rewritten as
	\begin{align}\label{eq:NDAE_with_UI_noise}
		\m Z\dot{{\m x}}\hspace{-0.1cm} &=\hspace{-0.1cm} {\m A}{\m x}\hspace{-0.1cm} + \hspace{-0.1cm} {\m F}{\m f}\left(\m x\right)\hspace{-0.1cm}+ \hspace{-0.1cm} \m B_u \Delta \m u\hspace{-0.1cm} +\hspace{-0.1cm} \m B_u \bar{\m u} \hspace{-0.1cm}+ \hspace{-0.1cm}\m B_q \bar{\m q} \notag 
		+{\m {B}_w} {{\m w} }\hspace{-0.1cm}+\hspace{-0.1cm} {\m H} \omega_{0}\\
		\m y &= \m C\m x  + {\m {D}_w} {{\m w} }.
	\end{align}

Motivated by \cite{WuIETIE2020,XuITVT2014,Sffker1995StateEO},  we propose a \textit{proportional integral} (\text{PI}) based framework to minimize the error arising due to the unknown inputs $\Delta \m u(t)$. The idea is basically that, instead of just estimating the system states through the observer we define $\Delta \m u(t)$ also as a state and then in observer design we determine two observer gains, \textit{(1)} Proportional gain which minimizes the error between original and estimated states, and \textit{(2)} Integral gain which compensate the estimation error dynamics for any inaccuracy caused by $\Delta \m u(t)$. Further details and working of PI-type Luenberger observers can be seen in \cite{WuIETIE2020,XuITVT2014,Sffker1995StateEO}. With that in mind the augmented NDAE, which is reformulation of \eqref{eq:NDAE_with_UI_noise} can be written as
\begin{align*}
	\begin{split}
		\underbrace{\bmat{{\m Z}&\m O \\ \m O&\m I}}_{\m Z_\varrho}\hspace{-0.05cm}\underbrace{\bmat{\dot{{\m x}} \\ {{\Delta \dot{\m u}}}}}_{\dot{{\m \varrho}}} \hspace{-0.1cm}&=\hspace{-0.1cm}\underbrace{\bmat{{\m A}&\m {B_u} \\ \m O &{\m O}}}_{\m {A_}\varrho}\hspace{-0.05cm}\underbrace{\bmat{{\m x} \\ {\Delta {\m u}}}}_{{\m \varrho}}\hspace{-0.05cm}+ 
		\hspace{-0.05cm}\underbrace{\bmat{{\m {F}} \\ \m O}}_{\m {F_}{\varrho}}\hspace{-0.05cm} {\m f(\m {x})}+
		\hspace{-0.05cm}\underbrace{\bmat{{\m B_u} \\ \m O}}_{\m {B}_{u\varrho}}\hspace{-0.05cm} { \bar{\m u}} \\&\,+ \hspace{-0.05cm}\underbrace{\bmat{{\m B_q} \\ \m O}}_{\m {B}_{q\varrho}}\hspace{-0.05cm} { \bar{\m q}} + \hspace{-0.05cm}\underbrace{\bmat{{\m {B_w}} \\ \m O}}_{ {\m B_{\m w,\varrho}}}\hspace{-0.05cm} {\m w} +\hspace{-0.05cm} \underbrace{\bmat{ {\m H} \\ \m O }}_{{\m H_\varrho}}\hspace{-0.05cm} \omega_0\\
		\m y &\,= \underbrace{\bmat{{\m C}&\m O}}_{\m C_\varrho}\hspace{-0.05cm}\underbrace{\bmat{{{\m x}} \\ {{\Delta {\m u}}}}}_{{{\m \varrho}}} +  \mD_w\m w.	 
	\end{split}
\end{align*}
The augmented system can be written in a compact form as follows
\begin{subequations}\label{eq:NDAE_PI}
	\begin{align*} 
		\begin{split}
			\m {Z}_\varrho\dot{{\m \varrho}} \hspace{-0.05cm}&=\hspace{-0.05cm} {\m {A}_\varrho}{{\m \varrho}}\hspace{-0.05cm} + \hspace{-0.05cm} {\m {F}_\varrho}{\m f}\left(\m x\right)\hspace{-0.05cm} +\hspace{-0.05cm} {\m {B}_{w,\varrho}}{{\m {w}} }\hspace{-0.05cm} \hspace{-0.05cm}+\hspace{-0.05cm}\m{B}_{q\varrho}{{\bar{\m q}}} +\m{B}_{u\varrho}{{\bar{\m u}}}+{\m {H}_\varrho} \omega_{0}
		\end{split}\\
		\begin{split}
			{\m y} \hspace{-0.05cm}&=\hspace{-0.05cm} \m C_{\varrho} {\m \varrho}+{{\m D_{w}}} {{\m w}}.
		\end{split} 
	\end{align*}
\end{subequations}
In the augmented system defined above the dynamics of unknown inputs are defined as $\Delta \dot{\m u}= \m{\psi}\Delta\m u$, where $\m{\psi}$ is a constant matrix and can be constructed based on the knowledge of unknown inputs. However, if the dynamics are completely unknown then $\m{\psi}= \m O$ can provide sufficient compensation to the estimation error dynamics \cite{Sffker1995StateEO}. \textcolor{black}{Notice that by choosing  $\m{\psi}= \m O$ we assume that the rate of change in the unknown inputs is negligible. }This allows the augmented system to be written as
\begin{align*}
	\begin{split}
		\underbrace{\bmat{{\m Z}&\m O \\ \m O&\m I}}_{\m Z_\varrho}\hspace{-0.05cm}\underbrace{\bmat{\dot{{\hat{\m x}}}\\{{\Delta \dot{\hat{\m u}}}}}}_{\dot{{\m \varrho}}} \hspace{-0.1cm}&=\hspace{-0.1cm}\underbrace{\bmat{{\m A}&\m {B}_u\\ \m O &{\m O}}}_{\m {A}_\varrho}\hspace{-0.05cm}\underbrace{\bmat{\hat{\m x} \\ \Delta{\hat{\m u}}}}_{{\m \varrho}}\hspace{-0.05cm}+ 
		\hspace{-0.05cm}\underbrace{\bmat{{\m {F}} \\ \m O}}_{\m {F}_{\varrho}}\hspace{-0.05cm} {\m f(\hat{ \m x})}+
		\hspace{-0.05cm}\underbrace{\bmat{{\m B_u} \\ \m O}}_{\m {B}_{u\varrho}}\hspace{-0.05cm} {\bar{\m u}} \\&\,+ \hspace{-0.05cm} \underbrace{\bmat{ {\m L_{\varrho,P}} \\ {\m L_{\varrho,I}}}}_{{\m L_\varrho}}\hspace{-0.05cm} (\m y - \hat {\m y}) + \hspace{-0.05cm}\underbrace{\bmat{{\m B_q} \\ \m O}}_{\m {B}_{q\varrho}}\hspace{-0.05cm} { \bar{\m q}} + \hspace{-0.05cm} \underbrace{\bmat{ {\m H} \\ \m O }}_{{\m H_\varrho}}\hspace{-0.05cm} \omega_0 \\
		\hat{\m y} &\,= \underbrace{\bmat{{\m C}&\m O}}_{\m C_\varrho}\hspace{-0.05cm}\underbrace{\bmat{{\hat{\m x}} \\ \Delta{\hat{\m u}}}}_{{{\m \varrho}}}	 
	\end{split}
\end{align*}
which can be simplified to
\begin{align}\label{eq:obsrDyn_PI}
	\begin{split}
	\hspace{-0.5cm}	\m {Z}_\varrho\dot{\hat {\m\varrho}}\hspace{-0.05cm} &=\hspace{-0.05cm} {\m {A}_\varrho}{\hat{\m \varrho}}\hspace{-0.05cm} + \hspace{-0.05cm} {\m {F}_\varrho}{\m f}\left(\hat{\m x}\right) \hspace{-0.05cm}+\hspace{-0.05cm}\m{B}_{u\varrho}{{\bar{\m u}}}\hspace{-0.05cm}+\hspace{-0.05cm} \m L_\varrho \Delta \m y \hspace{-0.05cm}
	+\hspace{-0.05cm}\m{B}_{q\varrho}{{\bar{\m q}}} \hspace{-0.05cm}+ 
	\hspace{-0.05cm}{\m H_\varrho} \omega_{0}\hspace{-0.4cm}\\
	\hat{\m y} &= \m C_\varrho \hat{\m \varrho}
	\end{split}
\end{align}
with $ \m e_\varrho = \m \varrho - \hat{\m \varrho}$ the estimation error dynamics can be derived in the same fashion as done in \eqref{eq:errorDyn} and can be written as
\begin{align}\label{eq:errorDyn_PI}
	\m Z_\varrho\dot{{\m e}}_\varrho &= \m A_{c,\varrho}\m e_\varrho + \m F_\varrho \Delta \m {f(x)}+ (\mB_{w,\varrho} - \mL_\varrho \mD_{w,\varrho})\m w
\end{align}
where $\m A_{c,\varrho}= (\m A_\varrho -\m L_\varrho \m C_\varrho)$. As the augmented error dynamics \eqref{eq:errorDyn_PI} retains the same structure as \eqref{eq:errorDyn_Hinf}, optimization problem $\mathbf{P_2}$ can be easily solved for the  observer dynamics presented in \eqref{eq:obsrDyn_PI}.
Note that the overall observer gain $\mL$ is now a combination of two gains $\mL_{e,P}$ and $\mL_{e,I}$ and it provides the following three main benefits: \textit{(1)} $\mL_{e,P}$ ensures that the performance of error  dynamics is robust against disturbances in $H_\infty$ sense as discussed in Section \ref{section:robust_obs_design}, \textit{(2)} $\mL_{e,P}$ also ensures that the error dynamics \eqref{eq:errorDyn_PI} converges asymptotically near to zero as $t\rightarrow\infty$, and \textcolor{black}{ \textit{(3)} $\mL_{e,I}$ compensates error dynamics \eqref{eq:errorDyn_PI} in the case when the observer does not know the exact values of control inputs. Or in other words $\mL_{e,I}$ try to remove any steady state error caused by the unknown inputs.}
\section{Case studies}\label{section:simulations}
The proposed approach is tested on western electricity coordinating council (WECC) system. The data for this test system has been taken from MATPOWER\cite{5491276} with a file name \texttt{case9}. To assess the performance and applicability of the proposed NDAE observer in tracking both dynamic and algebraic variables, DSE has been performed under different dynamic conditions and with changing loads and renewables. The simulations have been performed using MATLAB R$2021$b running on $64$-bit windows $10$ using Intel core i$9$-$11980$HK CPU with $64$GB RAM. Both the power system and observer NDAE models are solved using MATLAB index one DAEs solver \texttt{ode15i}, while all the convex SDPs optimizations are carried out in YALMIP \cite{LofbergICRA2004} using MOSEK \cite{Andersen2000} as a solver. To compute the Lipschitz matrix $\mG$, numerical method described in \cite[Section IV]{SebastianACC2019} is used to calculate Lipschitz constants for each nonlinearity in each state equation separately and then all of them are lumped in one diagonal matrix $\mG$, which is used throughout the case studies.   

For all the case studies the observer dynamics given in \eqref{eq:obsrDyn_PI} are simulated and the observer gain is calculated by solving optimization problem $\mathbf{P_2}$ with $c_1=1$, $c_2=1$ and $c_3=1/3$. All the states of the observer are initialized with random  initial conditions having $10\%$ maximum deviation from the steady state values of the power system except for generators speed which is kept the same as generator synchronous speed $w_0$. Three PMUs are installed at Buses $4$, $6$ and $8$, which are sufficient for the overall observability of WECC system.
The initial conditions and steady state values for the NDAE model are computed using power flow solution obtained from MATPOWER while the generators parameters are extracted from PST case file name \texttt{datane.m} \cite{sauer2017power}. Here we set $\omega_{0} = 2\pi60$ $\mr{{rad/sec}}$ and the power base is considered as $100$ MVA. Dynamic response of the power system has been achieved by introducing a fault at $\mr{t = 25sec }$ on line $4-9$ which is then cleared at $50$ $\mr{msec}$ and $200$ $\mr{msec}$ from near and remote end. 

\textcolor{black}{
The settings for \texttt{ode15i} is chosen to be: (1) relative tolerance $=1\times10^{-4}$ (2) absolute tolerance $=1\times10^{-5}$  and (3) maximum step size $=1\times10^{-3}$. Similarly for MOSEK we select: (1) SDP positive semidefinite constant = $1\times10^{-3}$, (2) maximum relative dual bound $1\times10^{-9}$, (3) maximum absolute dual bound $ = 1\times10^{-4}$, and (4) maximum absolute primal bound $ = 1\times10^{-5}$. The rest of the settings are kept to their default values.}
\vspace{-0.1cm}
\begin{algorithm}[t]
	\caption{\text{Implementation of the NDAE observer}}\label{alg:NDAE obsrvr}
	\DontPrintSemicolon 
	\textcolor{black}{Extract network description and generator parameters from MATPOWER and PST\;
		Generate matrices ${\m A}$, ${\m F}$, ${\m B}_u$, ${\m B}_q$, ${\m B}_w$ and ${\m H}$\;
		Select buses for the PMUs placements and compute $\m C$ matrix accordingly\;
		Create augmented state space model for both power system and observer dynamics\;
		Solve $\mathbf{P_2}$ and compute observer gain matrix $\m L_{\varrho}$\;
		Simulate power system under transient conditions,  initialize observer and use $\m L_{\varrho}$ to perform DSE}\;
\end{algorithm}
\vspace{-0.1cm}
\subsection{Case 1: DSE Under Gaussian and Non-Gaussian Noise}\label{Section:simulations_A}
We first analyze the performance of the observer under Gaussian process and measurement noise. We impose a Gaussian measurement noise having a diagonal covariance matrix and with error variance of $0.001^{2}$ and Gaussian process noise having also diagonal covariance matrix with each entry equal to $5\%$ of largest state change. Note that, there are no unknown inputs or load uncertainties and the observer knows the exact values of inputs and load demands in realtime. The results are presented in Fig. \ref{fig:Gaussian_Noises_dynamic}. For brevity, only Generator 1 state estimation results are shown. We can see that although the observer started from different initial conditions and is not aware of the Gaussian noise it is able to track all the states and thus making estimation error norm \eqref{eq:errorDyn_PI} asymptotically zero as shown in Fig. \ref{fig:error_norm}.

It is also important to analyze the performance of the observers under non-Gaussian noise, because in \cite{8067439} it has been shown that measurement noise cannot be always considered as Gaussian. Thus Cauchy noise has been generated by setting the noise vector $w_{mi} = a+b(\pi(R-0.5))$, where $a =0$, $b= 1\times 10^{-7}$ and $R$ is a random variable inside $(0,1)$. Instead of Gaussian noise, $w_{mi}$ has been added to the PMU measurements while the rest of all the setting for the observer has been kept the same. Similar results have been achieved as presented in Fig. \ref{fig:Gaussian_Noises_dynamic}. This can be validated from the estimation error norm given in Fig. \ref{fig:error_norm}. We can see that thanks to the $H_{\infty}$ based design the observer is still able to track all the original states and thus driving the error norm asymptotically to zero. 
\begin{figure}
	\includegraphics[keepaspectratio,scale=0.60]{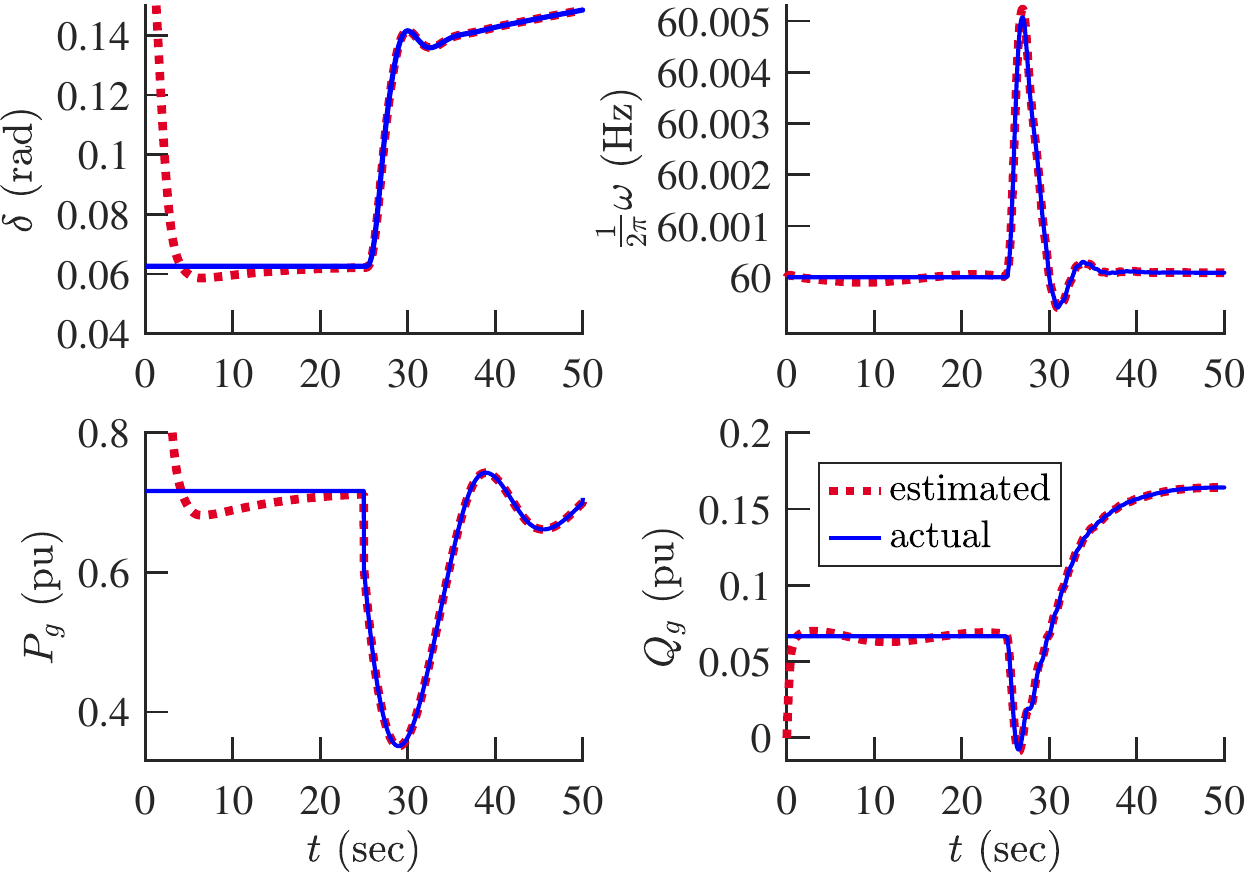}\vspace{-0.0cm}\caption{State estimation results for Generator 1 under Gaussian process and measurement noise.}
	\label{fig:Gaussian_Noises_dynamic}
\end{figure}
\vspace{-0.1cm}
\subsection{Case 2: DSE Under Unknown Control Inputs}\label{Section:simulations_C}	\vspace{-0.1cm}
For brushless synchronous generator's it is difficult to measure $T_{\mathrm{M}}$ and $E_{\mathrm{fd}}$. To that end in this section we consider the case when the observer is supplied only steady state values of the control inputs and the actual inputs are kept unknown to the observer. Also, Gaussian noise has been added to the system and PMUs measurement. Rest of all the settings for the observer and power network model are kept the same as in Section \ref{Section:simulations_A}. From Fig. \ref{fig:UIs} we can see that the observer is still able to provide accurate estimate of both both algebraic and dynamic  states. These observations can also be validated from the estimation error norm given Fig. \ref{fig:error_norm}.
\vspace{-0.1cm}
\subsection{Case 3: DSE Under Uncertainty from Loads/Renewables}\label{Section:simulations_D}
\vspace{-0.1cm}
In this section we discuss the performance of our observer with load and renewable disturbances. The simulations for this section have been performed as follows

Initially the power network operates with total load demand of $\left(P_{\mr{L}}^0+Q_{\mr{L}}^0\right)$ and total renewables power generation as $P_{\mr{R}}^0 = 0.2P_{\mr{L}}^0$. Then right after $t>0$ the power generated from renewables and total real system load demand experiences a step disturbances. There updated values are as $P_{\mr{R}}^d = (P_{\mr{R}}^0 + \Delta P_{\mr{R}}^0)$ and $P_{\mr{L}}^d = (P_{\mr{L}}^0 - \Delta P_{\mr{L}}^0)$, where we have set $\Delta P_{\mr{R}}^0 = 0.03P_{\mr{R}}^0$  and $\Delta P_{\mr{L}}^0 = 0.01P_{\mr{L}}^0$. Moreover, to account for random load and renewables variations, we also assume that these step disturbances contain noise, such that $P_{\mr{R}}^d = (P_{\mr{R}}^0 + \Delta P_{\mr{R}}^0) + q_r(t)$ and $P_{\mr{L}}^d = (P_{\mr{L}}^0 - \Delta P_{\mr{L}}^0) +q_l(t)$, where $q_r(t)$ and $q_l(t)$ are Gaussian noise with zero mean and variance of $0.002\Delta P_{\mr{R}}^0$ and $0.002\Delta P_{\mr{L}}^0$ respectively. All these disturbances are lumped in vector $\m w$. 

To perform DSE the initial conditions for the observer, number of PMUs, and level of process and measurement noise are kept the same as in Section \ref{Section:simulations_A}. The observer is only supplied with the steady state values of loads, thus the observer is kept completely unaware of the disturbances in loads and renewables and also process and measurement noise. Optimization problem $\mathbf{P_2}$ for observer dynamics \eqref{eq:obsrDyn_PI} has been solved. The estimation results are shown in Fig. \ref{fig:unkown loads and gaussian} while the estimation error norm is presented in Fig. \ref{fig:error_norm} and from which we can see that the observer is driving the error norm asymptotically near zero in less then $20s$.

\textcolor{black}{To further analyze the performance of the proposed observer against disturbances from load and renewables, we gradually increase (from $3\% - 15\%$) the amount of uncertainty from load and renewables and root
mean square error (RMSE) (Eq. \eqref{eq:RMSE}) between actual and estimated states has been recorded for each of the scenarios. We notice that for $3\%$ step disturbance from load and renewables the value of RMSE is $0.403$, similarly for $6\%, 8\%, 10\%$, and $15\%$ the RMSE values are $0.424$, $0.519$, $0.508$, and $0.619$ respectively. We can clearly see that by increasing the amount of uncertainty the value of RMSE still remains less than zero and thus the observer is providing good estimation results. This is because in $H_\infty$ based observer design the observer always tries to drive the nonlinear NDAE model of error dynamics \eqref{eq:errorDyn_Hinf} near zero without having any statistical knowledge about the uncertainty. This is one of the key advantage of the proposed methodology over Kalman filter based state estimation techniques.  } 
\begin{figure}
	\centering
	\includegraphics[keepaspectratio=true,scale=0.65]{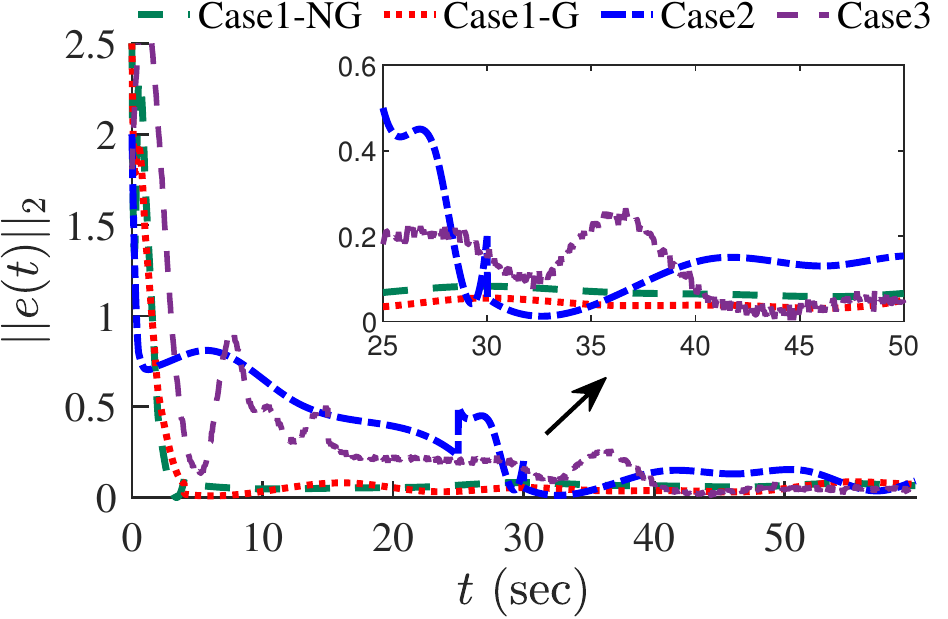}\vspace{-0.1cm}\caption{Comparison of estimation error norm for all the case studies. Case1-G and Case1-NG represent Case 1 with Gaussian and non-Gaussian noise respectively. }\label{fig:error_norm}
\end{figure}
\begin{figure}
		\centering
		\includegraphics[keepaspectratio=true,scale=0.60]{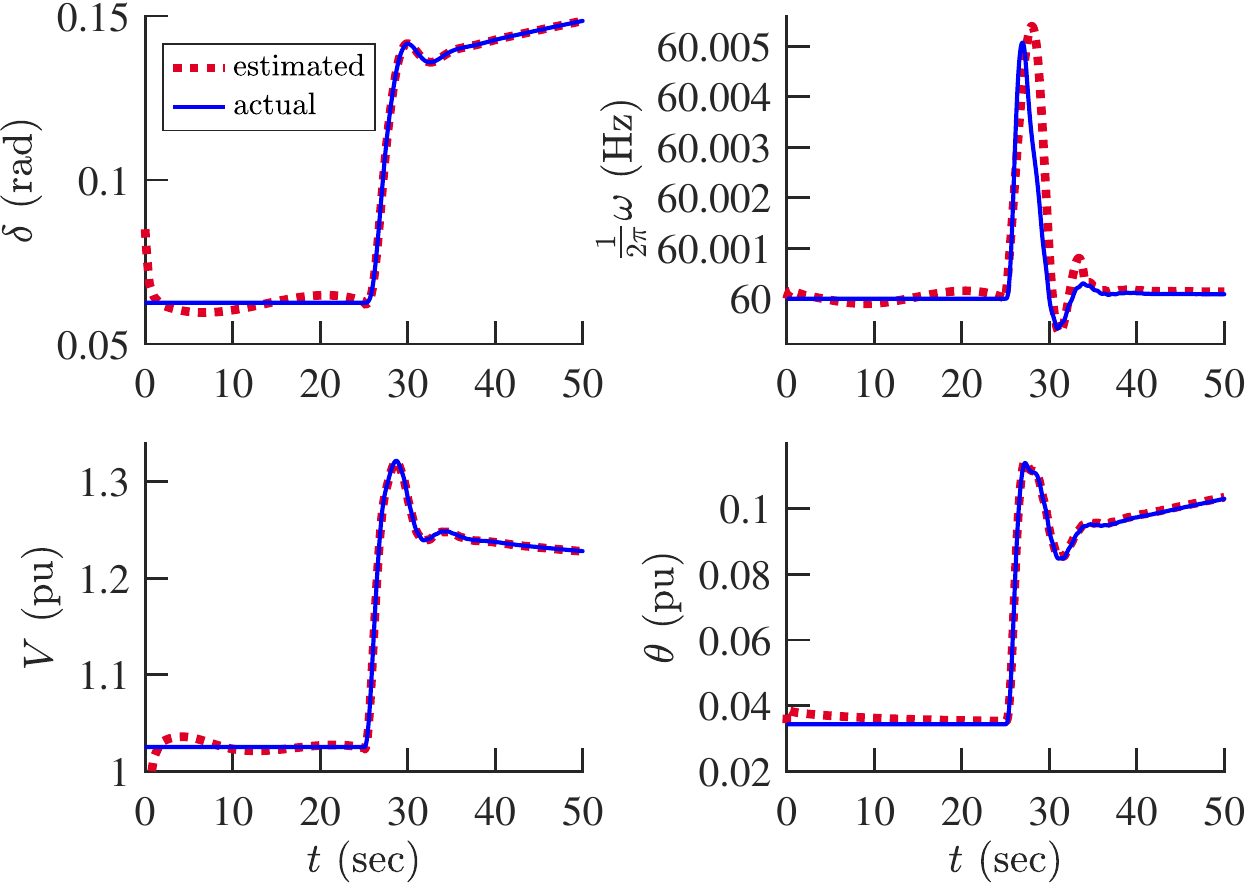}\vspace{-0.0cm}\caption{Estimates of dynamic states of Generator 1 and algebraic variables of Bus 6, with unknown control inputs.}\label{fig:UIs}
		\vspace{-0.1cm}
\end{figure}
\begin{figure}
	\centering 
	\subfloat{\includegraphics[keepaspectratio=true,scale=0.580]{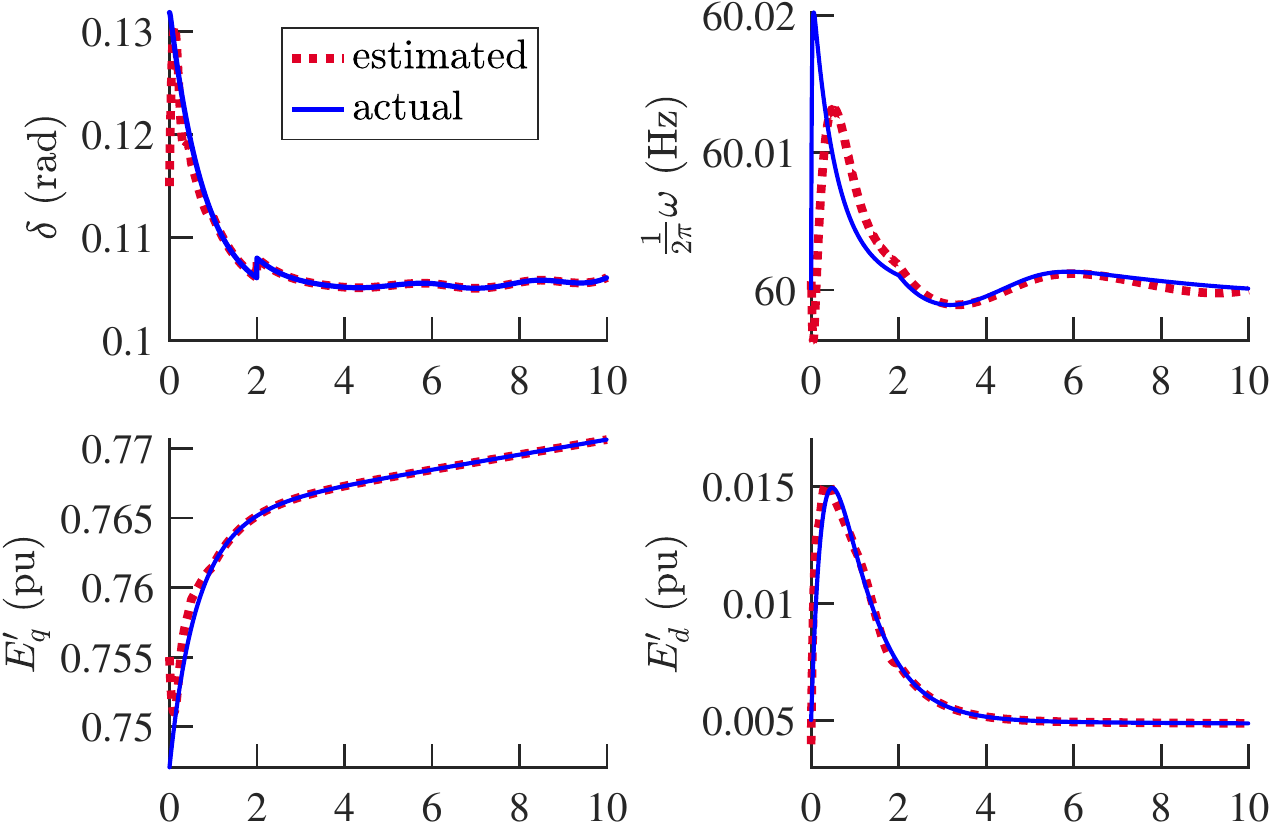}}{}{}\vspace{-0.1cm}
	\subfloat{\includegraphics[keepaspectratio=true,scale=0.580]{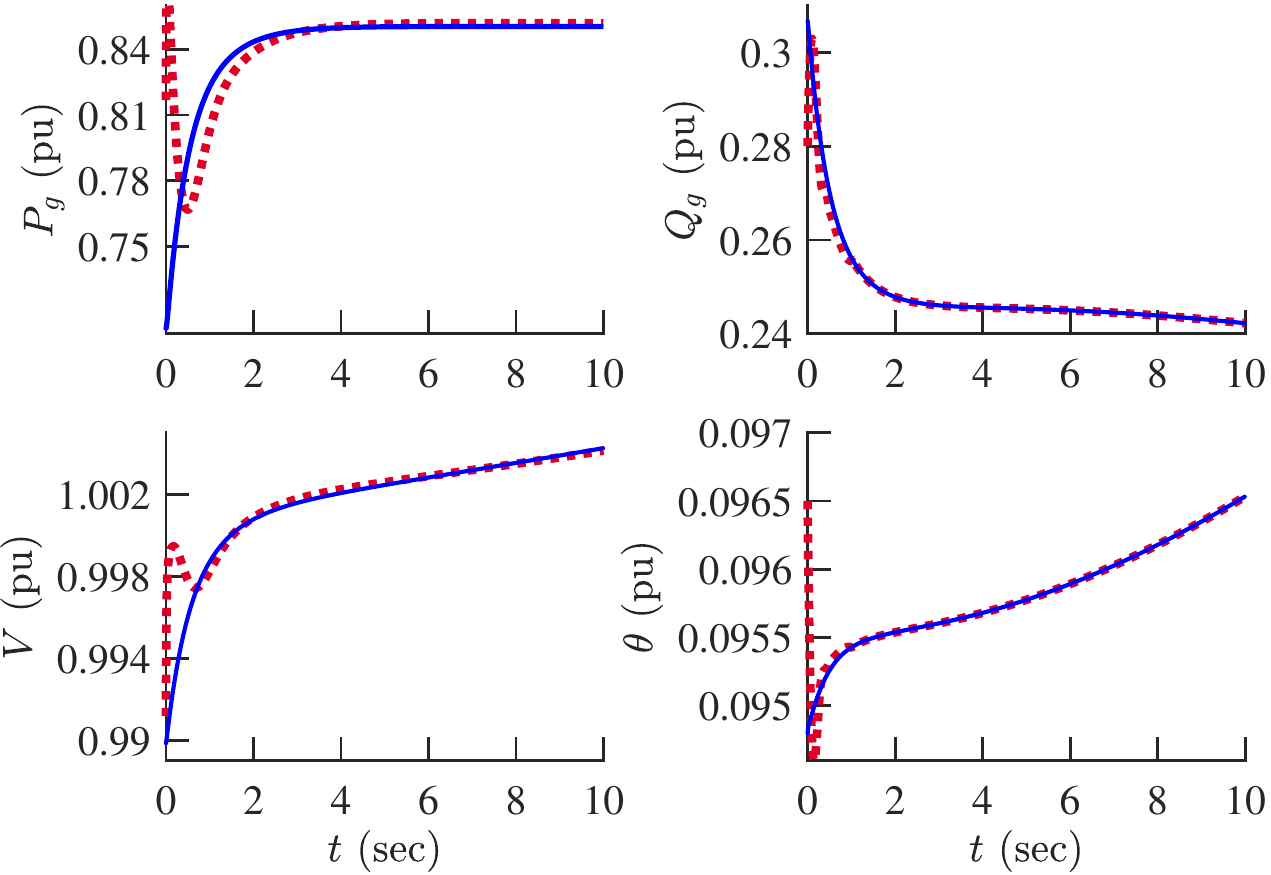}}{}{}\hspace{-0.1cm}\vspace{-0.1cm}
	\caption{Estimates of dynamic states of Generator 2 and algebraic variables of Bus 7 under loads and renewables uncertainties.}
	\label{fig:unkown loads and gaussian}
	\vspace{-0.3cm}
\end{figure}
\subsection{Comparison With the Decoupled Two Step Approach}
\vspace{-0.1cm}
In this section, we present the comparison of the NDAE observer proposed in this study with the two-stage techniques ---see \cite{ZhangITSE2014,RouhaniITSG2018,RinaldiINPROCEEDINGS2017} present in the literature. The two-step technique has been applied in this study as follows. In the first step least absolute value (LAV) method \cite{GolIETSG2014} based linear phasor estimator is used to approximate system algebraic variables $\m x_a$ and then in the second stage the estimated $\m x_a$ along with EKF/UKF/HEKF is used to estimate generator dynamic states $\m x_d$. Due to the decoupled two-step nature, this method can only be applied in a discrete time. Thus the  NDAE model \eqref{eq:final_DAE} is converted to discrete time model using first order Taylor approximation \cite{KAZANTZIS1999763}.

In this decoupled two-step approach (also referred to as LAV+EKF/UKF/HEKF) at each time step noisy PMUs measurements $\m y[k]$ are sampled and then the following optimization problem is carried out \cite{GolIETSG2014}
\begin{align*}
	\mathbf{\left( P_3\right) }\;\minimize_{{\m x_a},\m w_m} &\;\;\; \sum_{j=1}^p\norm {\m w_m}\\ \subjectto & \;\;\;\m y[k] = \tilde{\mC}\m x_a + \m w_m
\end{align*} 
where $\m x_a$ denote the algebraic variables and $\m w_m$ is the measurement noise. \textcolor{black}{After solving $\mathbf{P_3}$ a simple EKF \cite{Greg}, HEKF \cite{ZhaoITWPRS2018} and UKF \cite{GhahremaniITPWRS2011} has been implemented to estimate generator dynamic states. Mechanical torque $\m T_{\mathrm{M}}$ and field voltage $\mE_{\mathrm{fd}}$ are considered as inputs while the estimated algebraic variables from the first stage are also used to aid EKF, HEKF and UKF in estimating generator dynamics.}
 
To generate the dynamic response of the system a fault has been introduced at $t=11s$ on Line $5-6$. Two case studies are considered for the comparison. The first scenario is based on Case 1 (see Section \ref{Section:simulations_A}) in which Gaussian process and measurements noise are considered and it has been assumed that the accurate values of the control inputs $(T_{\mathrm{M}} $ and $E_{\mathrm{fd}})$ are available to the observer in realtime. The second case study is based on Case 2 (see Section \ref{Section:simulations_C}) where the observers have only been supplied the steady state values of the control inputs and the realtime transient values of the control inputs are kept unknown to the observers. The process and measurement noise covariance matrices and the rest of all the settings are kept the same as discussed in Case 1 and Case 2.
\vspace{-0.0cm}
\begin{table*}[t]
	\color{black}	\footnotesize	\centering 
	\caption{\textcolor{black}{Comparison of RMSE and computational time $\Delta t (s)$  for the NDAE observer and LAV+EKF/HEKF/UKF }}
	\begin{threeparttable}
		\begin{tabular}{p{8em}|c|c|c|c|c|c|c|c}
			\toprule
			\multirow{2}[4]{*}{DSE technique} & \multicolumn{2}{c|}{Case 1, Gaussian noise} & \multicolumn{2}{c|}{Case 1, non-Gaussian noise} & \multicolumn{2}{c|}{Case 2} & \multicolumn{2}{c}{Case 3} \\
			\cmidrule{2-9}    \multicolumn{1}{c|}{} & {$\Delta t$ (s)} & RMSE  & $\Delta t$ (s) & RMSE  & {$\Delta t$ (s)} & RMSE  & $\Delta t$ (s) & RMSE \\
			\midrule
			NDAE observer\tnote{$\dagger$} & 10.02319 & 0.03101 & 10.22451 & 0.06224 & 11.22403 & 0.1773 & 10.0701 & 0.3328 \\
			\midrule
			LAV+UKF & 251.9431 & 0.06324 & 237.3417 & 1.19413 & 249.7481 & 0.35261 & 249.4154 & 1.36917 \\
			\midrule
			LAV+EKF & 435.4125& 0.06391 & 465.3120 & 1.14389 & 439.8374 & 0.3605 & 440.1892 & 1.45101 \\
			\midrule
			\textcolor{black}{LAV+HEKF} & \textcolor{black}{438.4012}& \textcolor{black}{0.06411} & \textcolor{black}{480.2081} & \textcolor{black}{0.81089} & \textcolor{black}{465.9084} & \textcolor{black}{0.2691} & \textcolor{black}{410.9012} & \textcolor{black}{1.02310}\\
			\bottomrule
		\end{tabular}\label{tab:comp_time}%
		\begin{tablenotes}
			\item[$\dagger$] The computational time for the NDAE observer includes the time needed to calculate the observer gain matrix $\mL$. It took $5.408s$ for the optimization problem $\mathbf{P_2}$ to be solved in YALMIP with MOSEK as a solver. 
		\end{tablenotes}
	\end{threeparttable}
\vspace{-0.1cm}
\end{table*}%
\vspace{-0.0cm}
The simulation results are presented in Figs. \ref{fig:Comp_UI}, and \ref{fig:Comp_Gauss_delta}. \textcolor{black}{ We can see from Fig. \ref{fig:Comp_Gauss_delta} that in the first scenario  of the comparative study albeit all the observers can give relatively good estimation results, the performance of the NDAE observer is better as compared to the LAV+EKF/HEKF/UKF. We can also see that the performance of both  LAV+EKF/UKF is almost similar.} These observations are also supported from Fig. \ref{fig:Comp_error_norms_Gauss} where the estimation error norm for all the observers has been plotted, we can see that the NDAE observer has the least steady state error norm as compared to the $2-$stage observers.
\begin{figure}
	\centering 
	\subfloat{\includegraphics[keepaspectratio=true,scale=0.55]{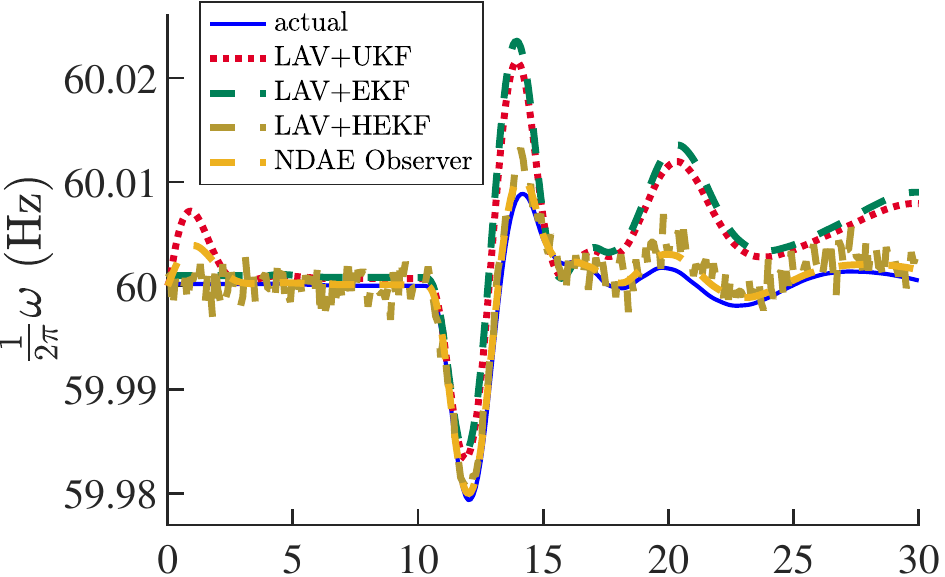}}{}{}\vspace{-0.1cm}
	\subfloat{\includegraphics[keepaspectratio=true,scale=0.55]{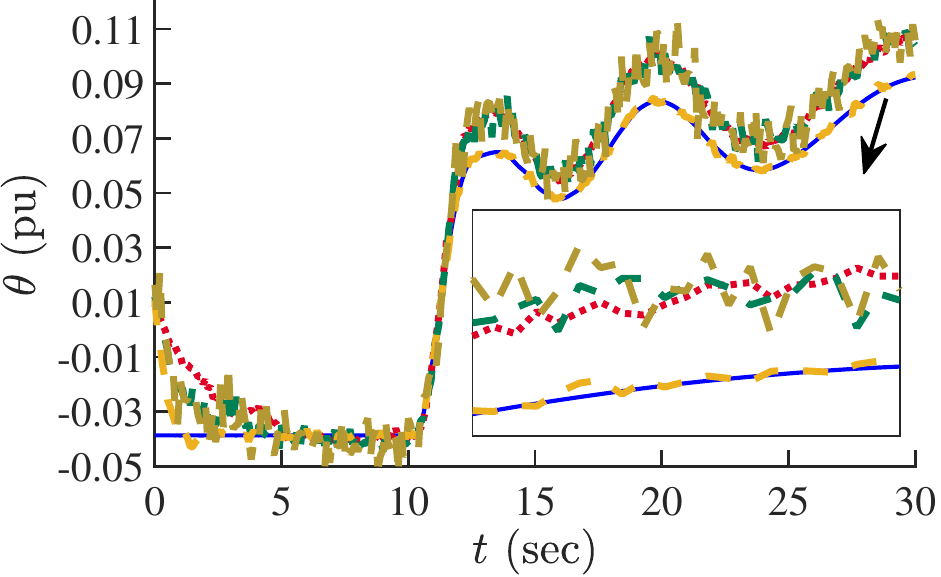}}{}{}\hspace{-0.1cm}\vspace{-0.3cm}
	\caption{\textcolor{black}{Comparison of estimates of Generator 3 frequency and Bus angle using LAV+EKF/UKF/HEKF and NDAE observer for Case 2.}}
	\label{fig:Comp_UI}
\end{figure}
For the second case study, we can see from Fig. \ref{fig:Comp_UI} that both LAV+EKF and LAV+UKF perform poorly during dynamic response of the system. The reason for that is because after the fault has been introduced in the system, both EKF and UKF do not know the transient parts of the control inputs and are only aware of the steady state values of $\m T_{\mathrm{M}}$ and $\mE_{\mathrm{fd}}$ and are thus giving poor estimates. 

On the other hand for the NDAE observer although the observer does not have the accurate knowledge of the control inputs it is still able to give accurate estimation results.  These results can also be corroborated by the estimation error norm given in Fig. \ref{fig:Comp_error_norms_UI} from which we can see that the NDAE observer is driving the error norm near zero and has the smallest steady-state value as compared to the $2-$stage observers.
\begin{figure}
	\centering 
	{\includegraphics[keepaspectratio=true,scale=0.65]{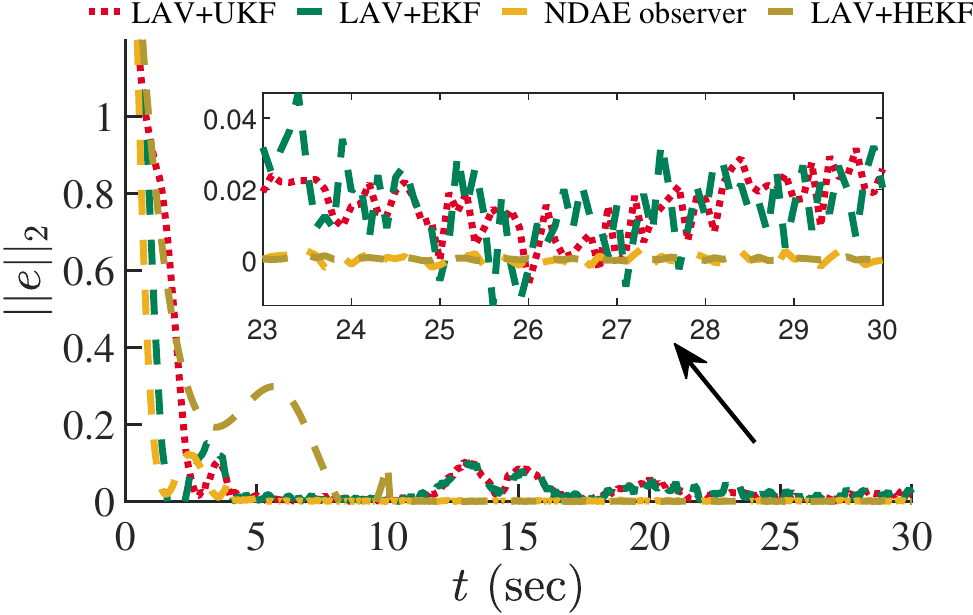}}{}{}\hspace{-0.1cm}
	\caption{\textcolor{black}{Comparison of the error norm for LAV+EKF/UKF/HEKF and NDAE observer for Case 1.}}
	\label{fig:Comp_error_norms_Gauss}
\end{figure} 

To further obtain a quantitative comparison between the LAV+EKF/UKF/HEKF and the proposed NDAE observer, root means square error (RMSE) and the overall computational time for all the observers has also been calculated and are presented in Tab. \ref{tab:comp_time}. The formula for RMSE is as follows
\begin{align}\label{eq:RMSE}
	\mathbf{RMSE} = \sum_{i=1}^n\sqrt{\dfrac{1}{t_f}\sum_{t=1}^{t_f}(x_i(t)-\hat{x}_i(t))^2}
\end{align}
\begin{figure}
	\centering 
	{\includegraphics[keepaspectratio=true,scale=0.65]{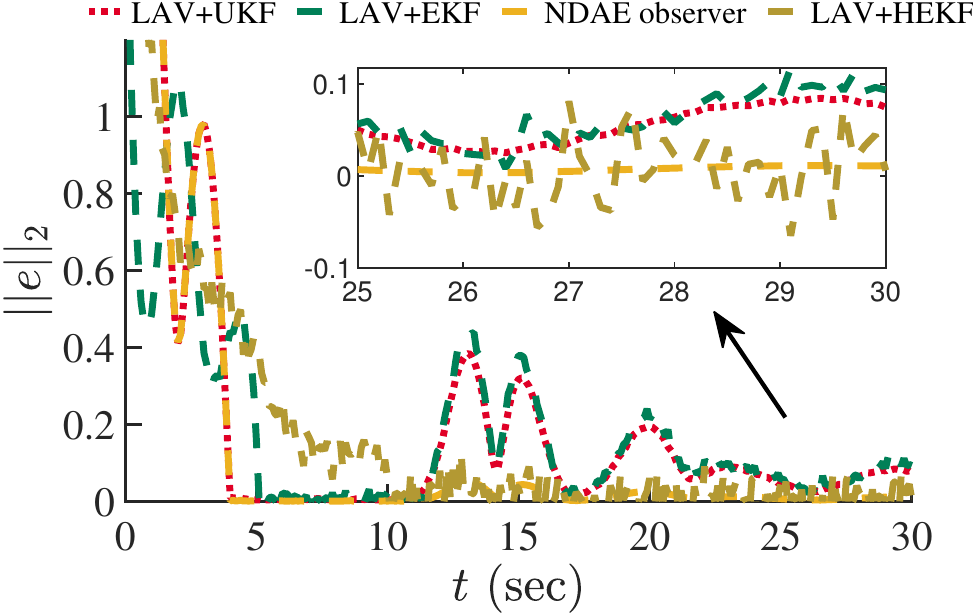}}{}{}\hspace{-0.1cm}
	\caption{\textcolor{black}{Comparison of the error norm for LAV+EKF/UKF/HEKF and NDAE observer for Case 2.}}
	\label{fig:Comp_error_norms_UI}
	\vspace{-0.3cm}
\end{figure}
From Tab. \ref{tab:comp_time} we can see that the NDAE observer is way more computationally efficient as compared to LAV+EKF/UKF/HEKF because it has the least computation time in performing the DSE followed by LAV+UKF while LAV+EKF/HEKF are the least efficient. This is because the NDAE observer is not recursive as compared to the stochastic estimators and the observer gain matrix is fixed. While on the other hand in the decoupled two-step approaches state estimation for both algebraic and dynamic states has to be performed separately and thus increases the computational complexity. Also in the second stage EKF/HEKF evaluates Jacobian in each iteration and thus it further increases computational time  while UKF utilizes unscented transformation and is thus a bit computationally more efficient than EKF/HEKF. 
\begin{figure}
	\centering 
	{\includegraphics[keepaspectratio=true,scale=0.65]{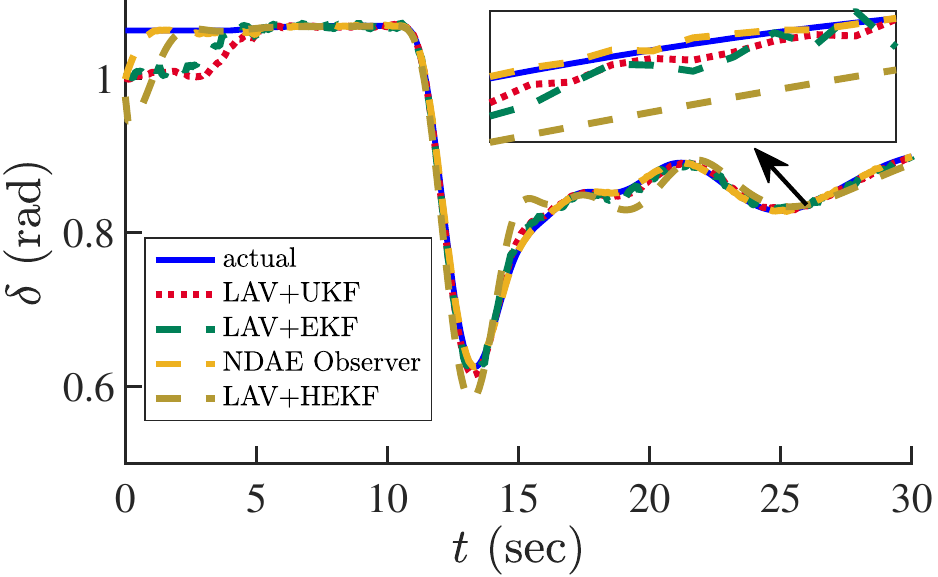}}{}{}\hspace{-0.1cm}\vspace{-0.1cm}
	\caption{\textcolor{black}{Estimation results of Generator 3 rotor angle using LAV+EKF/UKF/HEKF and NDAE observer for Case 1.}}
	\label{fig:Comp_Gauss_delta}
	\vspace{-0.3cm}
\end{figure} 
\vspace{-0.1cm}
\subsection{Extension to bigger system and $9^{th}$-order generator model}\label{section:bigger system}
\textcolor{black}{In the previous sections we showcase the performance of the proposed observer on a WECC 3-machines 9-bus system with a fourth-order generator model. The considered test power system in the earlier sections is relatively small and does not include the models of prime mover and generator's excitation system (exciter, stabilizer, and amplifier dynamics). For better accuracy and to fully capture the dynamics of synchronous generators after a fault, a higher-order ($6^{th}$ or $9^{th}$) generator model should be used.
With that in mind, and to further advocate the feasibility of the proposed NDAE observer for practical applications, here we use IEEE 10-machine 39-bus system with a standard $9^{th}$-order generator model to perform DSE \cite{sauer2017power}. The overall dynamical model now consists of IEEE Type DC1 excitation system and steam/hydro turbine dynamics as given in \cite{sauer2017power}.}  

\textcolor{black}{With that in mind, the overall states, control inputs, and output vectors for this model can be written as}
\textcolor{black}{\begin{align*}
	\hspace{-0.2cm}\m x \hspace{-0.1cm}&=\hspace{-0.1cm} \bmat{\m\delta^\top\,\,\m \omega^\top\,\,\m E'^\top_{q}\,\,\mE'^\top_{d}\,\,\m P_{v}^\top\,\,\m T_{\mathrm{M}}^\top\,\,\m E_{\mathrm{fd}}^\top\,\,\m r_{f}^\top\,\,\m v_{a}^\top\,\,\m V_{R}^\top\,\,\m V_{I}^\top\,\,\m I_R^\top\,\,\m I_I^\top}\\
	\m y \hspace{-0.1cm}&= \hspace{-0.1cm}\bmat{\m{\mr V}^\top&\m{\mr I}^\top}^\top, \m u = \bmat{\m V_{\mathrm{ref}}^\top & \m P_{v_{\mathrm{ref}}}^\top}
\end{align*}}
\textcolor{black}{where $\m P_{v}$ denotes steam/hydro valve position, $\m T_{\mathrm{M}}$ represents prime mover torque, $\m V_{R}, \m V_{I}, \m I_R, \m I_I$ are the real and imaginary parts of voltage and current phasors of all the buses, and $\m E_{\mathrm{fd}}$, $\m v_{a}$, and $\m r_{f}$ depicts field voltage, amplifier voltage, and stabilizer output respectively. Similarly $\m V_{\mathrm{ref}}$ and $\m P_{v_{\mathrm{ref}}}$ represent voltage reference point and valve position set points for the synchronous generator's.} 

\textcolor{black}{The simulation studies have been carried out before, during, and after the disturbances for $30s$. To ensure the observability of IEEE 39-bus system we follow \cite{ChakrabartiITPWRS2008} and deploy 13 PMUs on buses $[2,6,9,10,13,14,17,19,20,22,23,25,29]$. To calculate the NDAE observer gain matrix $\m L$, again we solve optimization problem $\mathbf{P_2}$ in YALMIP with MOSEK as solver. All the settings for MOSEK and \texttt{ode15i} solver are kept the same as discussed in the previous test case system. Two case studies are considered for this simulation test:} 
\begin{itemize}
\item \textcolor{black}{Case 4}: \textcolor{black}{For this case a disturbance has been applied after $3s$ and we impose Gaussian measurement noise having a diagonal covariance matrix and variance of $0.01^{2}$ and Gaussian process noise having also diagonal covariance matrix with each entry equal to $5\%$ of largest state change.}
	\item \textcolor{black}{Case 5}: \textcolor{black}{For this case study, right after $t>0$ we apply step disturbances (similar to as done in Case 3)  to the renewables and overall load demand of the system while also considering Gaussian process and measurement noise. The observer is only supplied with the steady-state values of loads/renewables and thus the observer is kept completely unaware of the disturbances in loads and renewables.}
\end{itemize}
\textcolor{black}{For both Case 4 and Case 5 the NDAE observer has been initialized from random initial conditions having $10\%$ maximum deviation from steady state values. After performing DSE we obtain the value of RMSE for Case 4 as 5.278, similarly for Case 5 we obtain 9.248. Furthermore, the estimation results are shown in Fig. \ref{fig:case4} and \ref{fig:case5} while the estimation error norm for both case studies is presented in Fig. \ref{fig:error ieee39}. We can clearly see that although the observer is completely unaware of disturbances in load and renewables and process/measurement noise, it can still drive the estimation error norm near zero and thus provide good estimation results.}
\begin{figure}[h]
	{\includegraphics[keepaspectratio=true,scale=0.507]{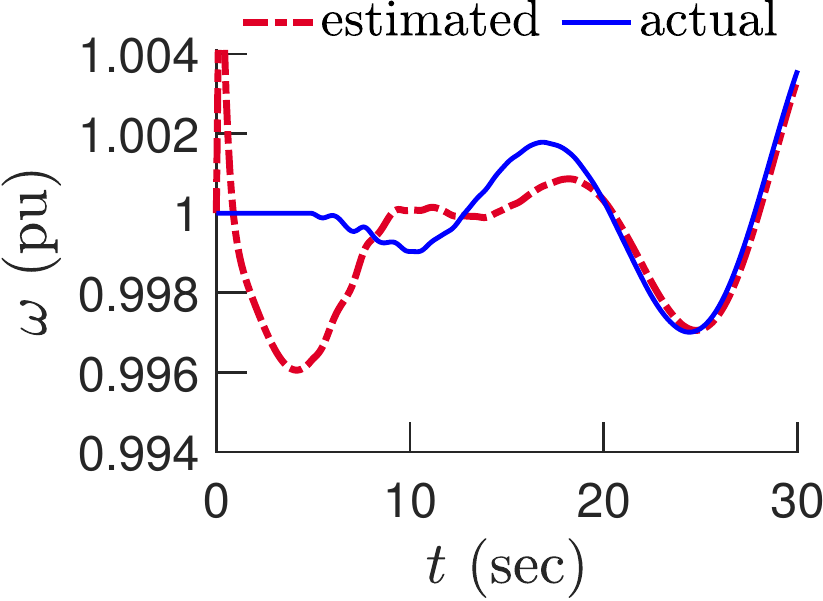}}{}
	{\includegraphics[keepaspectratio=true,scale=0.507]{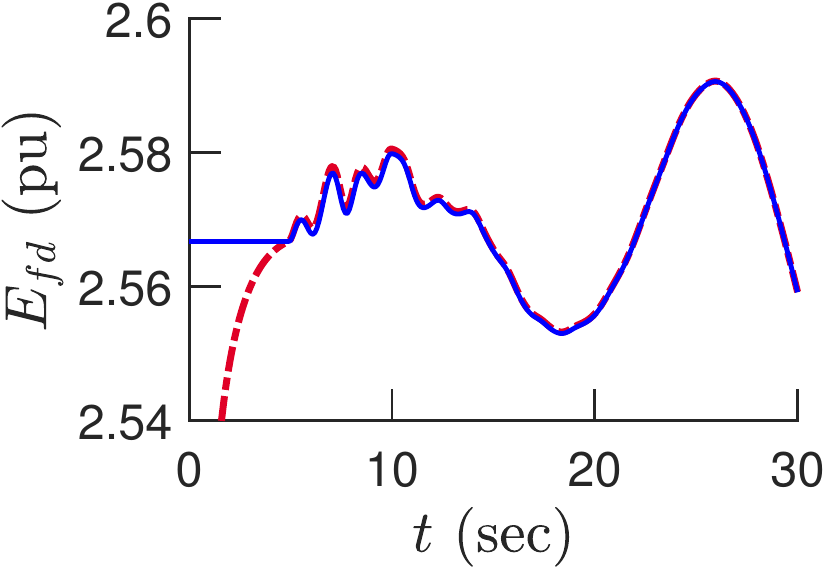}}{}\hspace{0.01cm}
	\caption{\textcolor{black}{Estimation results for IEEE 39-bus system considering Case 4, rotor speed and field voltage of Generator 4.}}\label{fig:case4}
\end{figure}

\begin{figure}[h]
	{\includegraphics[keepaspectratio=true,scale=0.507]{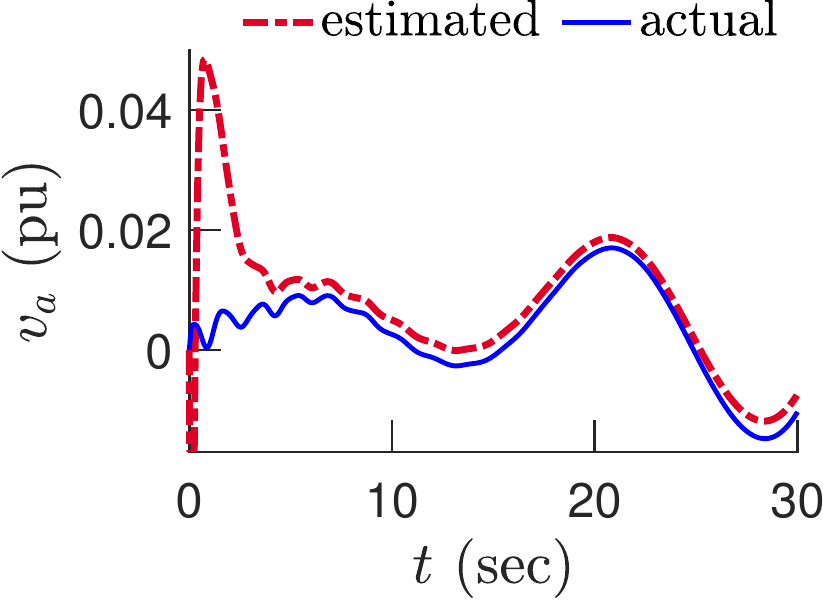}}{}\hspace{0.1cm}
	{\includegraphics[keepaspectratio=true,scale=0.507]{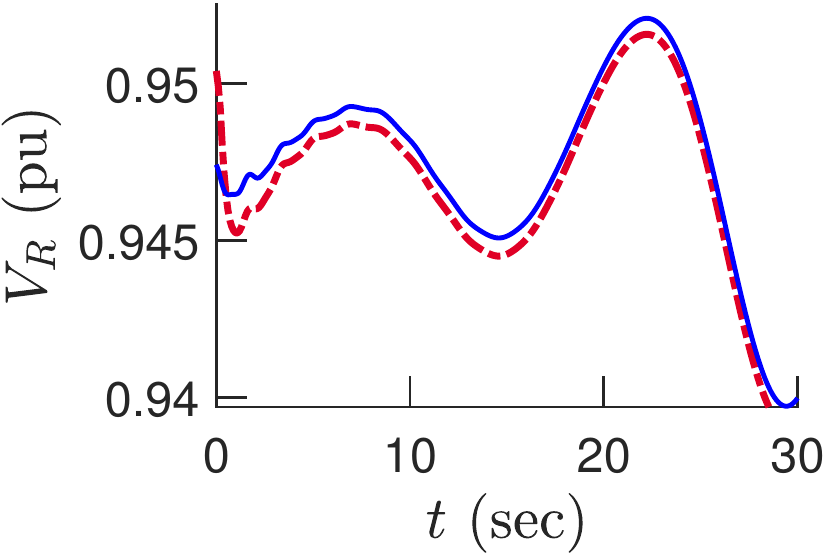}}{}
	\caption{\textcolor{black}{Estimation results for IEEE 39-bus system considering Case 5, amplifier voltage of Generator 5 and real voltage at Bus 19.}}\label{fig:case5}
\end{figure}

\begin{figure}[h]
	\centering 
	{\includegraphics[keepaspectratio=true,scale=0.55]{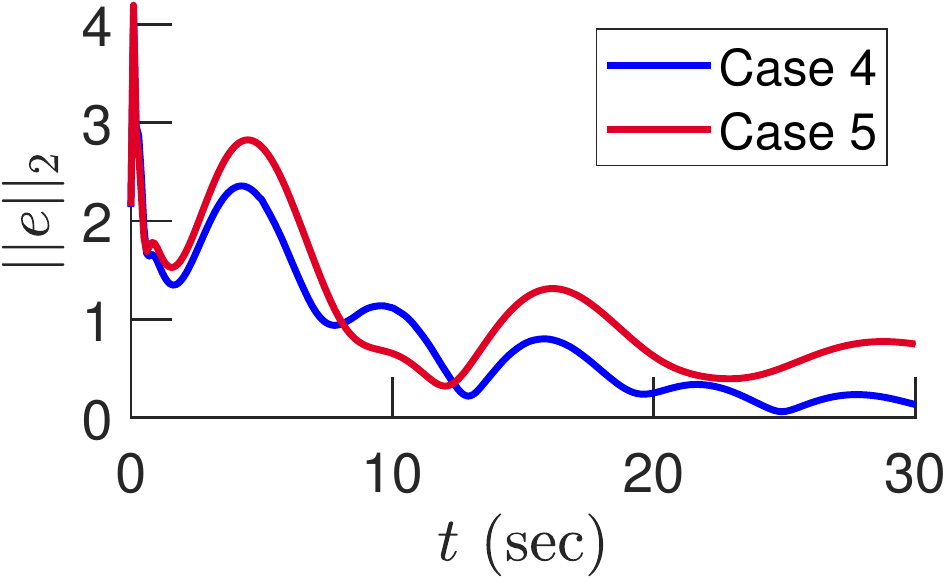}}{}{}\hspace{-0.1cm}
	\caption{\textcolor{black}{Estimation error norm for Case 4 and Case 5.}}
	\label{fig:error ieee39}
	\vspace{-0.3cm}
\end{figure} 

\section{Conclusion paper limitations and future work}\label{section:conclusion}
In this paper, a novel estimator is proposed to simultaneously measure both algebraic and dynamic states of the NDAE model of power systems using PMU measurements. As compared to the observers present in literature the proposed observer: \textit{(1)} does not require any statistical properties of the disturbances, \textit{(2)} has no assumptions on PMUs locations and only require that the overall system should be observable, and \textit{(3)} is non-recursive and the observer gain remains fixed and thus more feasible for practical implementations. A thorough simulations studies showcase the performance of the observer under different dynamic conditions and with varying load and renewables. The simulations results show that the NDAE observer give better estimates of both algebraic and dynamic states and is computationally more efficient as compared to the 2-stage observers. 

\color{black}

Although the advantages of our NDAE observer are evident, the proposed methodology still has several limitations.  First, in this work we have considered renewables such as solar PV and wind farms as negative real power load and we assumed that the variations in renewables are only affecting the real power demand of the system. Further modeling of renewables is hence needed. Second, this work is based on centralized state estimations. Thus, to compute observer gain matrix $\m L$ for a very large power system the SDP solver can take more time to solve  $\mathbf{P_2}$. 
In particular, solving $\mathbf{P_2}$ for a large power system can be tedious and the solver may take more time to determine optimal gain matrix $\m L$. With that in mind, we point out here that $\mathbf{P_2}$ is solved \textit{offline} not online---the observer gain $\m L$ is \textit{time-invariant}. Once appropriate gain matrix is calculated it can be
	used offline and the observer can simultaneously estimate all
	the states of the power system in a few seconds (this is the time required by the ode15i solver to solve the observer dynamics) using a few PMU measurements. The proposed observer is far less
	computationally intensive because once the gain is determined
	the observer essentially works as a one-step predictor. Consequently, DSE
	can be performed with a significantly smaller number of matrix
	multiplications and hence lower computational time as demonstrated in the case studies.  

\normalcolor

 Third, the theory of the observer design is based on a continuous-time model of synchronous machine and PMU measurements, however, the measurements from PMUs are usually transmitted via a digitalized transmission system. Thus a discrete-time version of the proposed NDAE observer would be more suitable as compared to this one. To that end as a future research work, we plan to include the actual converter-based models of the renewable energy resources along with synchronous generator dynamics. Finally, the proposed NDAE observer will be extended to a discretized and decentralized framework to perform DSE for a very large power system.
\vspace{-0.1cm}
\bibliographystyle{IEEEtran} 
\bibliography{mybibfile}

\begin{thebibliography}{10}
\providecommand{\url}[1]{#1}
\csname url@samestyle\endcsname
\providecommand{\newblock}{\relax}
\providecommand{\bibinfo}[2]{#2}
\providecommand{\BIBentrySTDinterwordspacing}{\spaceskip=0pt\relax}
\providecommand{\BIBentryALTinterwordstretchfactor}{4}
\providecommand{\BIBentryALTinterwordspacing}{\spaceskip=\fontdimen2\font plus
\BIBentryALTinterwordstretchfactor\fontdimen3\font minus
  \fontdimen4\font\relax}
\providecommand{\BIBforeignlanguage}[2]{{%
\expandafter\ifx\csname l@#1\endcsname\relax
\typeout{** WARNING: IEEEtran.bst: No hyphenation pattern has been}%
\typeout{** loaded for the language `#1'. Using the pattern for}%
\typeout{** the default language instead.}%
\else
\language=\csname l@#1\endcsname
\fi
#2}}
\providecommand{\BIBdecl}{\relax}
\BIBdecl

\bibitem{sauer2017power}
P.~Sauer, M.~Pai, and J.~Chow, \emph{Power System Dynamics and Stability: With
  Synchrophasor Measurement and Power System Toolbox}, ser. Wiley - IEEE.\hskip
  1em plus 0.5em minus 0.4em\relax Wiley, 2017.

\bibitem{kundur2007power}
P.~Kundur, ``Power system stability,'' \emph{Power system stability and
  control}, pp. 7--1, 2007.

\bibitem{Liu2021TPWRS}
Y.~Liu, A.~K. Singh, J.~Zhao, A.~P.~S. Meliopoulos, B.~C. Pal, M.~A. Bin
  Mohd~Ariff, T.~Van~Cutsem, M.~Glavic, Z.~Huang, I.~Kamwa, L.~Mili, A.~S. Mir,
  A.~F. Taha, V.~Terziya, and S.~Yu, ``Dynamic state estimation for power
  system control and protection,'' \emph{IEEE Transactions on Power Systems},
  pp. 1--1, 2021.

\bibitem{8624411}
J.~Zhao, A.~Gómez-Expósito, M.~Netto, L.~Mili, A.~Abur, V.~Terzija, I.~Kamwa,
  B.~Pal, A.~K. Singh, J.~Qi, Z.~Huang, and A.~P.~S. Meliopoulos, ``Power
  system dynamic state estimation: Motivations, definitions, methodologies, and
  future work,'' \emph{IEEE Transactions on Power Systems}, vol.~34, no.~4, pp.
  3188--3198, 2019.

\bibitem{ZhenyuIPEC2007}
Z.~Huang, K.~Schneider, and J.~Nieplocha, ``Feasibility studies of applying
  kalman filter techniques to power system dynamic state estimation,'' in
  \emph{2007 International Power Engineering Conference (IPEC 2007)}, 2007, pp.
  376--382.

\bibitem{GhahremaniITPWRS2011}
E.~Ghahremani and I.~Kamwa, ``Online state estimation of a synchronous
  generator using unscented kalman filter from phasor measurements units,''
  \emph{IEEE Transactions on Energy Conversion}, vol.~26, no.~4, pp.
  1099--1108, 2011.

\bibitem{CuiITWPRS2015}
Y.~Cui and R.~Kavasseri, ``A particle filter for dynamic state estimation in
  multi-machine systems with detailed models,'' \emph{IEEE Transactions on
  Power Systems}, vol.~30, no.~6, pp. 3377--3385, 2015.

\bibitem{ZhouITPWRS2013}
N.~Zhou, D.~Meng, and S.~Lu, ``Estimation of the dynamic states of synchronous
  machines using an extended particle filter,'' \emph{IEEE Transactions on
  Power Systems}, vol.~28, no.~4, pp. 4152--4161, 2013.

\bibitem{6281499}
Y.~Li, Z.~Huang, N.~Zhou, B.~Lee, R.~Diao, and P.~Du, ``Application of ensemble
  kalman filter in power system state tracking and sensitivity analysis,'' in
  \emph{PES T D 2012}, 2012, pp. 1--8.

\bibitem{ZhouITSG2015}
N.~Zhou, D.~Meng, Z.~Huang, and G.~Welch, ``Dynamic state estimation of a
  synchronous machine using pmu data: A comparative study,'' \emph{IEEE
  Transactions on Smart Grid}, vol.~6, no.~1, pp. 450--460, 2015.

\bibitem{LiuIEEACC2020}
H.~Liu, F.~Hu, J.~Su, X.~Wei, and R.~Qin, ``Comparisons on kalman-filter-based
  dynamic state estimation algorithms of power systems,'' \emph{IEEE Access},
  vol.~8, pp. 51\,035--51\,043, 2020.

\bibitem{GhahremaniITWPRS2016}
E.~Ghahremani and I.~Kamwa, ``Local and wide-area pmu-based decentralized
  dynamic state estimation in multi-machine power systems,'' \emph{IEEE
  Transactions on Power Systems}, vol.~31, no.~1, pp. 547--562, 2016.

\bibitem{ZhaoITWPRS2017}
J.~Zhao, M.~Netto, and L.~Mili, ``A robust iterated extended kalman filter for
  power system dynamic state estimation,'' \emph{IEEE Transactions on Power
  Systems}, vol.~32, no.~4, pp. 3205--3216, 2017.

\bibitem{ZhaoITWPRS2018}
J.~Zhao, ``Dynamic state estimation with model uncertainties using
  {$\mathcal{H}_{\infty}$} extended kalman filter,'' \emph{IEEE Transactions on
  Power Systems}, vol.~33, no.~1, pp. 1099--1100, 2018.

\bibitem{ZhaoITSG2019}
J.~Zhao and L.~Mili, ``Robust unscented kalman filter for power system dynamic
  state estimation with unknown noise statistics,'' \emph{IEEE Transactions on
  Smart Grid}, vol.~10, no.~2, pp. 1215--1224, 2019.

\bibitem{LiIEEACC2019}
Y.~Li, J.~Li, J.~Qi, and L.~Chen, ``Robust cubature kalman filter for dynamic
  state estimation of synchronous machines under unknown measurement noise
  statistics,'' \emph{IEEE Access}, vol.~7, pp. 29\,139--29\,148, 2019.

\bibitem{SebastianITPWRS2020}
S.~A. Nugroho, A.~F. Taha, and J.~Qi, ``Robust dynamic state estimation of
  synchronous machines with asymptotic state estimation error performance
  guarantees,'' \emph{IEEE Transactions on Power Systems}, vol.~35, no.~3, pp.
  1923--1935, 2020.

\bibitem{Jin2018MultiplierbasedOD}
M.~Jin, H.~Feng, and J.~Lavaei, ``Multiplier-based observer design for
  large-scale lipschitz systems,'' 2018.

\bibitem{HaesITSG2020}
H.~Haes~Alhelou, M.~E. Hamedani~Golshan, and N.~D. Hatziargyriou,
  ``Deterministic dynamic state estimation-based optimal lfc for interconnected
  power systems using unknown input observer,'' \emph{IEEE Transactions on
  Smart Grid}, vol.~11, no.~2, pp. 1582--1592, 2020.

\bibitem{GRO201612}
\BIBentryALTinterwordspacing
T.~Groß, S.~Trenn, and A.~Wirsen, ``Solvability and stability of a power
  system dae model,'' \emph{Systems and Control Letters}, vol.~97, pp. 12--17,
  2016. [Online]. Available:
  \url{https://www.sciencedirect.com/science/article/pii/S0167691116301098}
\BIBentrySTDinterwordspacing

\bibitem{Wu2019InfluenceOL}
D.~Wu and B.~Wang, ``Influence of load models on equilibria, stability and
  algebraic manifolds of power system differential-algebraic system,''
  \emph{2019 57th Annual Allerton Conference on Communication, Control, and
  Computing (Allerton)}, pp. 787--795, 2019.

\bibitem{RouhaniITSG2018}
A.~Rouhani and A.~Abur, ``Linear phasor estimator assisted dynamic state
  estimation,'' \emph{IEEE Transactions on Smart Grid}, vol.~9, no.~1, pp.
  211--219, 2018.

\bibitem{ZhangITSE2014}
J.~Zhang, G.~Welch, G.~Bishop, and Z.~Huang, ``A two-stage kalman filter
  approach for robust and real-time power system state estimation,'' \emph{IEEE
  Transactions on Sustainable Energy}, vol.~5, no.~2, pp. 629--636, 2014.

\bibitem{RinaldiINPROCEEDINGS2017}
G.~Rinaldi and A.~Ferrara, ``Higher order sliding mode observers and nonlinear
  algebraic estimators for state tracking in power networks,'' in \emph{2017
  IEEE 56th Annual Conference on Decision and Control (CDC)}, 2017, pp.
  6033--6038.

\bibitem{RouhaniITPWRS2017}
A.~Rouhani and A.~Abur, ``Observability analysis for dynamic state estimation
  of synchronous machines,'' \emph{IEEE Transactions on Power Systems},
  vol.~32, no.~4, pp. 3168--3175, 2017.

\bibitem{ZhengITPWRS2021}
Z.~Zheng, Y.~Xu, L.~Mili, Z.~Liu, M.~Korkali, and Y.~Wang, ``Observability
  analysis of a power system stochastic dynamic model using a derivative-free
  approach,'' \emph{IEEE Transactions on Power Systems}, vol.~36, no.~6, pp.
  5834--5845, 2021.

\bibitem{QiITPWRS2015}
J.~Qi, K.~Sun, and W.~Kang, ``Optimal pmu placement for power system dynamic
  state estimation by using empirical observability gramian,'' \emph{IEEE
  Transactions on Power Systems}, vol.~30, no.~4, pp. 2041--2054, 2015.

\bibitem{ThabetITCST2018}
A.~Thabet, G.~B.~H. Frej, and M.~Boutayeb, ``Observer-based feedback
  stabilization for lipschitz nonlinear systems with extension to
  {$\mathcal{H}_{\infty}$} performance analysis: Design and experimental
  results,'' \emph{IEEE Transactions on Control Systems Technology}, vol.~26,
  no.~1, pp. 321--328, 2018.

\bibitem{ChenITAC2007}
M.-S. Chen and C.-C. Chen, ``Robust nonlinear observer for lipschitz nonlinear
  systems subject to disturbances,'' \emph{IEEE Transactions on Automatic
  Control}, vol.~52, no.~12, pp. 2365--2369, 2007.

\bibitem{GuopingITCS2006}
G.~Lu and D.~Ho, ``Full-order and reduced-order observers for lipschitz
  descriptor systems: the unified lmi approach,'' \emph{IEEE Transactions on
  Circuits and Systems II: Express Briefs}, vol.~53, no.~7, pp. 563--567, 2006.

\bibitem{AlessandriITAC2020}
A.~Alessandri and F.~Boem, ``State observers for systems subject to bounded
  disturbances using quadratic boundedness,'' \emph{IEEE Transactions on
  Automatic Control}, vol.~65, no.~12, pp. 5352--5359, 2020.

\bibitem{PhamICTSL2019}
T.-P. Pham, O.~Sename, and L.~Dugard, ``Unified {$\mathcal{H}_{\infty}$}
  observer for a class of nonlinear lipschitz systems: Application to a real er
  automotive suspension,'' \emph{IEEE Control Systems Letters}, vol.~3, no.~4,
  pp. 817--822, 2019.

\bibitem{9735348}
S.~A. Nugroho, A.~Taha, N.~Gatsis, and J.~Zhao, ``Observers for differential
  algebraic equation models of power networks: Jointly estimating dynamic and
  algebraic states,'' \emph{IEEE Transactions on Control of Network Systems},
  pp. 1--1, 2022.

\bibitem{ZhaoITPWRS2020}
J.~Zhao, Z.~Zheng, S.~Wang, R.~Huang, T.~Bi, L.~Mili, and Z.~Huang,
  ``Correlation-aided robust decentralized dynamic state estimation of power
  systems with unknown control inputs,'' \emph{IEEE Transactions on Power
  Systems}, vol.~35, no.~3, pp. 2443--2451, 2020.

\bibitem{QiIETPWRS2017}
J.~Qi, J.~Wang, H.~Liu, and A.~D. Dimitrovski, ``Nonlinear model reduction in
  power systems by balancing of empirical controllability and observability
  covariances,'' \emph{IEEE Transactions on Power Systems}, vol.~32, no.~1, pp.
  114--126, 2017.

\bibitem{SebastianACC2019}
S.~A. Nugroho, A.~F. Taha, and J.~Qi, ``Characterizing the nonlinearity of
  power system generator models,'' in \emph{2019 American Control Conference
  (ACC)}, 2019, pp. 1936--1941.

\bibitem{ShakedITAC1990}
U.~Shaked, ``{$\mathcal{H}_{\infty}$} minimum error state estimation of linear
  stationary processes,'' \emph{IEEE Transactions on Automatic Control},
  vol.~35, no.~5, pp. 554--558, 1990.

\bibitem{Andersen2000}
E.~D. Andersen and K.~D. Andersen, \emph{The Mosek Interior Point Optimizer for
  Linear Programming: An Implementation of the Homogeneous Algorithm}.\hskip
  1em plus 0.5em minus 0.4em\relax Boston, MA: Springer US, 2000, pp. 197--232.

\bibitem{Slemma}
P.~M.~C. Derinkuyu, Kurşad, ``On the s-procedure and some variants,''
  \emph{Mathematical Methods of Operations Research}, vol.~64, no.~1, pp.
  55--77, 2016.

\bibitem{LofbergICRA2004}
J.~Lofberg, ``Yalmip : a toolbox for modeling and optimization in matlab,'' in
  \emph{2004 IEEE International Conference on Robotics and Automation (IEEE
  Cat. No.04CH37508)}, 2004, pp. 284--289.

\bibitem{WuIETIE2020}
M.~Wu, F.~Gao, P.~Yu, J.~She, and W.~Cao, ``Improve disturbance-rejection
  performance for an equivalent-input-disturbance-based control system by
  incorporating a proportional-integral observer,'' \emph{IEEE Transactions on
  Industrial Electronics}, vol.~67, no.~2, pp. 1254--1260, 2020.

\bibitem{XuITVT2014}
J.~Xu, C.~C. Mi, B.~Cao, J.~Deng, Z.~Chen, and S.~Li, ``The state of charge
  estimation of lithium-ion batteries based on a proportional-integral
  observer,'' \emph{IEEE Transactions on Vehicular Technology}, vol.~63, no.~4,
  pp. 1614--1621, 2014.

\bibitem{Sffker1995StateEO}
D.~S{\"o}ffker, T.~Yu, and P.~M{\"u}ller, ``State estimation of dynamical
  systems with nonlinearities by using proportional-integral observer,''
  \emph{International Journal of Systems Science}, vol.~26, pp. 1571--1582,
  1995.

\bibitem{5491276}
R.~D. Zimmerman, C.~E. Murillo-Sánchez, and R.~J. Thomas, ``Matpower:
  Steady-state operations, planning, and analysis tools for power systems
  research and education,'' \emph{IEEE Transactions on Power Systems}, vol.~26,
  no.~1, pp. 12--19, 2011.

\bibitem{8067439}
S.~Wang, J.~Zhao, Z.~Huang, and R.~Diao, ``Assessing gaussian assumption of pmu
  measurement error using field data,'' \emph{IEEE Transactions on Power
  Delivery}, vol.~33, no.~6, pp. 3233--3236, 2018.

\bibitem{GolIETSG2014}
M.~Gol and A.~Abur, ``Lav based robust state estimation for systems measured by
  pmus,'' \emph{IEEE Transactions on Smart Grid}, vol.~5, no.~4, pp.
  1808--1814, 2014.

\bibitem{KAZANTZIS1999763}
\BIBentryALTinterwordspacing
N.~Kazantzis and C.~Kravaris, ``Time-discretization of nonlinear control
  systems via taylor methods,'' \emph{Computers Chemical Engineering}, vol.~23,
  no.~6, pp. 763--784, 1999. [Online]. Available:
  \url{https://www.sciencedirect.com/science/article/pii/S0098135499000071}
\BIBentrySTDinterwordspacing

\bibitem{Greg}
G.~Welch and G.~Bishop, ``An introduction to the kalman filter,''
  \url{https://www.cs.unc.edu/~welch/media/pdf/kalman_intro.pdf}, 2004.

\bibitem{ChakrabartiITPWRS2008}
S.~Chakrabarti and E.~Kyriakides, ``Optimal placement of phasor measurement
  units for power system observability,'' \emph{IEEE Transactions on Power
  Systems}, vol.~23, no.~3, pp. 1433--1440, 2008.

\bibitem{RisbudIPESISGT2016}
P.~Risbud, N.~Gatsis, and A.~Taha, ``Assessing power system state estimation
  accuracy with gps-spoofed pmu measurements,'' in \emph{2016 IEEE Power Energy
  Society Innovative Smart Grid Technologies Conference (ISGT)}, 2016, pp.
  1--5.

\bibitem{ZimmermanITPWRS2011}
R.~D. Zimmerman, C.~E. Murillo-Sánchez, and R.~J. Thomas, ``Matpower:
  Steady-state operations, planning, and analysis tools for power systems
  research and education,'' \emph{IEEE Transactions on Power Systems}, vol.~26,
  no.~1, pp. 12--19, 2011.

\end{thebibliography}
\vspace{-0.1cm}
\appendices
\vspace{-0.3cm}
\section{PMU Measurement Model}\label{Appndx:pmu measurement}
\vspace{-0.3cm}
In this work we are considering a more realistic PMUs placement model and thus there is no requirement for the PMUs to be placed at a particular location. Instead, we are assuming that the PMUs are already placed optimally in the network to make sure the whole system is observable. With that in mind let us define $\mc{N}_P \in \mc{N} $ as the set of buses connected with PMUs, then PMU at bus $j\in \mc{N}_P$ measures voltage phasors of that bus, ${\m{\mr V}} = \bmat{\{v_{Rj}\}_{j\in \mathcal N} +\{v_{Ij}\}_{j\in \mathcal N}}$  and current phasors, $\m {\mr I} = \bmat{\{I_{Rji}\}_{i\in \mathcal N_j} +\{I_{Iji}\}_{i\in \mathcal N_j}}$of bus $j$ and nearby buses $(\mathcal N_j \in \mc N)$ connected with bus $j$. Note that we can write the current phasor $\m I$ in term of $\m{\mr v}$ by using the below formula  \cite{RisbudIPESISGT2016}
\begin{align}\label{eq:PMU1}
	\begin{split}
		\bmat{\m I_R\\\m I_I} = \bmat{\mr {Re} \m Y_{ft}&\mr {-Im} \m Y_{ft}\\\mr {Im} \m Y_{ft}& \mr {Re} \m Y_{ft} }\m {\mr v}^\top
	\end{split}
\end{align}
where $\m Y_{ft}$ is the branch admittance matrix and can be calculated as $\m Y_{ft}= \bmat{\m Y_{f}\\\m -Y_{t}}$. $\m Y_{t}$ and $\m Y_{f}$ are the $to$ and $from$ branch admittance matrices. These matrices can be easily extracted from load flow analysis in  MATPOWER\cite{ZimmermanITPWRS2011}. The total PMU measurement model can be written as  $ \m \tilde{\m y}=\tilde{\m C}\m x_a$\footnote{The bus voltages in the NDAE model \eqref{eq:nonlinearDAEexplicit-1} are in polar coordinates, which have been converted to rectangular coordinates for using formula \eqref{eq:PMU1} and to find PMU measurement outputs. MATLAB function $\mr pol2cart$ has been used for this purpose in simulations.} with $\tilde{\m C}$  as 
\begin{align}\label{eq:PMU_C_matrix}
	\begin{split}
		\tilde{\m C} =  \bmat{\m O&\m O&\m s_n&\m O\\\m O&\m O&\m O&\m s_n\\\m O&\m O& \m S_n \mr R_e \left(\m Y_{ft} \right)& \m {-S}_n \mr I_m \left(\m Y_{ft} \right) \\\m O&\m O& \m S_n \mr I_m \left(\m Y_{ft} \right) & \m S_n \mr R_e \left(\m Y_{ft} \right) }
	\end{split}
\end{align}
where $\m s_n \in \mc N \times 1$ is a binary selection vector which has $1$ only at those locations where PMUs are placed and zeroes otherwise. Similarly $ \m S_n \in |\mc N_j|\times|2\mathcal{E}|$ is a binary selection matrix which has $1$ only at those lines which are originating from buses connected with PMUs. Recall the structure of state variable $\m x$, then the PMUs measurement model in term of $\m x$ can be written as
\vspace{-0.2cm}
\begin{align}\label{eq:PMU_final_model}
	\m y = \underbrace{\bmat{\m O & \tilde{\m C}}}_{\m C} \bmat{\m x_d\\ \m x_a}, \m y = \m C\m x
\end{align}
where the vector $\m y\in\mbb{R}^p$ collects all the measurements from PMUs and $\mC\in\mbb{R}^{p\times n}$ is the overall output matrix.
\vspace{-0.1cm}
\section{Details of Matrices Used in NDAE Model \eqref{eq:nonlinearDAE}}\label{appdx:A}
\vspace{-0.1cm}
The matrix ${\m A}_d$ is created as
\begin{align*}
	\bmat{\m O &\m I & \m O & \m O \\ \m O & -\Diag\left(\m D \oslash \m M\right)& \m O & \m O \\\m O &\m O & -{\m A}_{d(3,3)} & \m O \\ \m O & {\m O} &\m O & -\Diag\left(\m 1\oslash \m T'_{\mr{q0}}\right)}
\end{align*}
where ${\m A}_{d(3,3)}$ is given as
\begin{align*}
	{\m A}_{d(3,3)} &= \Diag\left(\m x_{\mr{d}}\oslash \left(\m x_{\mr{d}}'\odot \m T'_{\mr{d0}}\right) \right)
\end{align*}
similarly ${\m F}_d$ is specified as
\begin{align*}
	{\m F}_d & = \bmat{\m O&\m O& \m O \\-\Diag\left(\m 1\oslash \m M\right)&\m O & \m O\\ \m O& \m G_{d(3,2)} &\m O\\ \m O&\m O&\m G_{d(4,3)}} \\
\end{align*}
where $ \m F_d(3,2)$ and $ \m F_d(4,3)$ is given as
\begin{align*}
\m F_{d(3,2)} = \Diag\left((\m x_{\mr{d}}-\m x_{\mr{d}}')\oslash \left(\m x_{\mr{d}}'\odot \m T'_{\mr{d0}}\right) \right) \\
\m F_{d(4,3)} = \Diag\left((\m x_{\mr{q}}-\m x_{\mr{q}}')\oslash \left(\m x_{\mr{q}}'\odot \m T'_{\mr{q0}}\right) \right) 
\end{align*}
The function ${\m f}_d(\cdot)$ and ${\m f}_a(\cdot)$ in \eqref{eq:nonlinearDAE-1} and \eqref{eq:nonlinearDAE-2} are given as
\begin{align*}
	{\m f}_d\left({\m x}\right) = \bmat{\m P_{\mr{G}}\\ \{v_i\cos(\delta_{i}-\theta_i)\}_{i\in \mc{G}} \\\{v_i\sin(\delta_{i}-\theta_i)\}_{i\in \mc{G}}}
\end{align*}
\begin{align*}
	{\m f}_a\hspace{-0.05cm}\left({\m x}\right) \hspace{-0.075cm} =\hspace{-0.075cm} \bmat{\{E'_{qi}v_i\sin(\delta_i-\theta_i)\}_{i\in \mc{G}} \\\{v_i^2\sin(2(\delta_i-\theta_i))\}_{i\in \mc{G}} \\\{E'_{qi}v_i\cos(\delta_i-\theta_i)\}_{i\in \mc{G}} \\ \{v_i^2\}_{i\in \mc{G}} \\ \{v_i^2\cos(2(\delta_i-\theta_i))\}_{i\in \mc{G}} \\ \left\{ \hspace{-0.05cm}\sum_{j=1}^{N}\hspace{-0.05cm} v_iv_j\hspace{-0.05cm}\left(G_{ij}\cos \theta_{ij} \hspace{-0.05cm}+ \hspace{-0.05cm}B_{ij}\sin \theta_{ij}\right)\right\}_{i\in \mc{G}\hspace{-0.05cm}} \\ \left\{\hspace{-0.05cm}\sum_{j=1}^{N}\hspace{-0.05cm} v_iv_j\hspace{-0.05cm}\left(G_{ij}\sin \theta_{ij} \hspace{-0.05cm}- \hspace{-0.05cm}B_{ij}\cos \theta_{ij}\right)\right\}_{i\in \mc{G}\hspace{-0.05cm}} \\
		\left\{ \hspace{-0.05cm}\sum_{j=1}^{N}\hspace{-0.05cm} v_iv_j\hspace{-0.05cm}\left(G_{ij}\cos \theta_{ij} \hspace{-0.05cm}+ \hspace{-0.05cm}B_{ij}\sin \theta_{ij}\right)\right\}_{i\in \mc{N}\setminus\mc{G}\hspace{-0.05cm}} \\ \left\{\hspace{-0.05cm}\sum_{j=1}^{N}\hspace{-0.05cm} v_iv_j\hspace{-0.05cm}\left(G_{ij}\sin \theta_{ij} \hspace{-0.05cm}- \hspace{-0.05cm}B_{ij}\cos \theta_{ij}\right)\right\}_{i\in \mc{N}\setminus\mc{G}\hspace{-0.05cm}}}\hspace{-0.1cm}
\end{align*} 
Next, the matrix $\m F_a = \Blkdiag\left(\tilde{\m F}_{a},\m I\right)$ where $\tilde{\m F}_{a}$ is given as
\begin{align*}
	\tilde{\m F}_a & = \bmat{\Diag\left(\m 1\oslash \m x_{\mr{d}}'\right)&\m O\\\tilde{\m F}_{a(1,2)}&\m O\\\m O&\Diag\left(\m 1\oslash \m x_{\mr{d}}'\right)\\\m O&-\tilde{\m F}_{a(2,4)}\\ \m O & \tilde{\m F}_{a(2,5)}}^\top
\end{align*}
with $\tilde{\m F}_{a(1,2)} = \Diag\left((\m x_{\mr{d}}'-\m x_{\mr{q}})\oslash \left(2\m x_{\mr{d}}'\odot \m x_{\mr{q}}\right) \right)$,  $ \tilde{\m F}_{a(2,4)} = \Diag\left((\m x_{\mr{d}}'+\m x_{\mr{q}})\oslash \left(2\m x_{\mr{d}}'\odot \m x_{\mr{q}}\right) \right)$ and $ \tilde{\m F}_{a(2,5)} = \tilde{\m F}_{a(1,2)}$. 

The matrices $\mA_a$ and $\mB_a$ are given as
\begin{align*}
	\m A_a = \bmat{-\m I&\m O\\ {\m A}_p&\m O},\;\m B_a = \bmat{\m O\\ {\m B}_p}
\end{align*} 
where ${\m A}_p = \bmat{-\m I \;\;\;\m O}^\top$ and ${\m B}_p$ is a binary selection matrix having $1's$ at those location where buses are connected to loads and/or renewables.  Finally, the matrices $\m B_d$ and $\m h$ are constructed as follows
\begin{align*}
	{\m B}_d &= \bmat{\m O&\m O\\\m O&\Diag\left(\m 1\oslash \m M\right)\\\Diag\left(\m 1\oslash \m T'_{{\mr{d0}}}\right)&\m O\\\m O&\m O}, {\m h} =\bmat{-\m 1 \\ \m D \oslash \m M \\ \m O \\ \m O}.
\end{align*}
\end{document}